%% file: main.tex
\documentclass[12pt,a4paper]{article}

\PassOptionsToPackage{dvipsnames}{xcolor}

\usepackage{flonat-wp}
\usepackage[harvard]{flonat-bib}
\usepackage{xurl}
\hypersetup{hypertexnames=false}

\providecommand{\M}{\mathcal{M}}
\providecommand{\C}{\mathcal{C}}
\providecommand{\U}{\mathcal{U}}
\providecommand{\ind}{\mathbbm{1}}
\DeclareMathOperator{\Sh}{Sh}

\providecommand{\new}[1]{#1}
\providecommand{\removed}[1]{}

\title{Quotient Semivalues for False-Name-Resistant Data Attribution}

\author{%
  Florian A.\ D.\ Burnat\thanks{School of Management, University of Bath, UK. \texttt{fadb20@bath.ac.uk}}%
  \and%
  Brittany I.\ Davidson\thanks{School of Management, University of Bath, UK. \texttt{bid23@bath.ac.uk}}%
}

\date{29 July 2026}

\begin{document}

\begin{titlepage}
    \maketitle
    \begin{abstract}
        \noindent
Data valuation methods allocate payments and audit training data's contribution to machine-learning pipelines; however, they often assume passive contributors. In reality, contributors can split datasets across pseudonymous identities, duplicate high-value examples, create near-duplicates, or launder synthetic variants to inflate their share. We formalize this as false-name manipulation in ML data attribution. Our main construction is the \emph{quotient semivalue} mechanism: compute Shapley-, Banzhaf-, or Beta-style values over evidence-backed attribution clusters instead of raw identities, using a canonical-representative operator to absorb within-cluster duplication. We prove an impossibility: on a fixed monotone data-value game, exact Shapley-fair attribution over reported identities is incompatible with unrestricted false-name-proofness, even on binary-valued instances, and characterize the split-gain of a general semivalue on a unanimity counter-example. The mechanism is exactly false-name-proof under two structural conditions: \emph{false-name-neutral} within-cluster allocation and \emph{quotient-stable} manipulations. Under imperfect provenance, manipulation gain is bounded by allocation leakage, escaped-cluster mass, matched-cluster semivalue drift, and value-estimation error; cluster-level fairness loss is bounded separately by quotient-game distance. We instantiate the mechanisms in \textsc{DataMarket-Gym}, a benchmark for attribution under strategic provider attacks. On synthetic classification tasks, quotient semivalues with example-level evidence reduce manipulation gain on duplicate and near-duplicate Sybil attacks from $1.74$ under baseline Shapley to $0.96$, near the honest level. The cosine-threshold and (false-merge, false-split) rate sweeps trace the corresponding fairness--Sybil frontier. Image (CIFAR-10, frozen ResNet-18) and text (AG~News, MiniLM) experiments show the mechanism's behavior in real domains and reveal a domain-specific threshold reversal driven by embedding-space scale.

    \end{abstract}
    \setcounter{page}{0}
    \thispagestyle{empty}
\end{titlepage}

\clearpage
\onehalfspacing
\setlength{\emergencystretch}{2em}

\input{sections/01-introduction}
\input{sections/02-related-work}
\input{sections/03-strategic-model}
\input{sections/04-shapley-limits}
\input{sections/05-quotient-mechanisms}
\input{sections/06-guarantees}
\input{sections/07-experiments}
\input{sections/08-discussion}

\clearpage
\printbibliography

\appendix
\input{sections/appendix-a-proofs}
\input{sections/appendix-b-algorithms}
\input{sections/appendix-c-benchmark}
\input{sections/appendix-d-allocation}
\input{sections/appendix-e-results}
\input{sections/appendix-f-reproducibility}

\end{document}

%% file: sections/01-introduction.tex
\section{Introduction}\label{sec:intro}

A payment mechanism for training data does not merely measure value; it creates incentives about what data is submitted, how it is partitioned, and under whose identity it appears. The standard experimental setting for data valuation hides this: a fixed benchmark dataset, fixed unit decomposition, and fixed player set; however, in a real training-data market, the unit of attribution is endogenous. A provider who owns one dataset can submit it once, split it across accounts, duplicate selected examples, produce near-duplicate paraphrases or augmentations, or curate examples to exploit the evaluator's validation set. Once attribution scores determine money, access, reputation, or licensing credit, providers have reason to optimize the attribution mechanism rather than the social value of their data.

We study this at the interface of ML data valuation and false-name-proof mechanism design. Data Shapley \parencite{Ghorbani2019-vs} and its successors --- more efficient \parencite{Jia2019-ku}, more stable under noisy utility \parencite{Kwon2022-yc,Wang2023-cl}, and connected to the core \parencite{Yan2021-zs} --- divide the value among fixed, passive units. In contrast, false-name manipulation is the classical failure mode in combinatorial auctions and distributed systems \parencite{Douceur2002-ok,Yokoo2004-do,Conitzer2010-is}, in which a single agent benefits by entering under multiple identities. Data markets inherit this with the additional complication that the object is non-rival and that learning utility is non-additive in providers \parencite{Agarwal2019-jb,Acemoglu2022-ck}. The combination is the central claim: \emph{Shapley-fair data attribution is not automatically incentive-compatible}. If the reported game changes when a provider splits, the sum of pseudonym Shapley values can exceed the honest provider's Shapley value --- a structural issue stemming from the same complementarity that makes data valuable.

We make four contributions: (i) a model formalizing data-provider manipulation as mechanism design over latent providers, submitted identities, learning utility, and manipulation classes; (ii) an impossibility --- no mechanism is exactly Shapley-fair and unrestricted false-name-proof; (iii) \emph{quotient semivalue} mechanisms that compute semivalues over evidence-backed clusters with a canonical-representative operator, with manipulation gain decomposed into allocation leakage, escaped mass, matched-cluster semivalue drift, and estimation error; and (iv) \textsc{DataMarket-Gym}, a benchmark covering exact-duplicate, near-duplicate, synthetic-laundering, Sybil-splitting, strategic-curation, and poisoning attacks scored by utility, fairness loss, manipulation gain, and runtime.

The positive claim is conditional but useful: when similarity and provenance signals can reliably collapse manipulated variants into the same attribution unit, semivalue payment recovers a controlled trade-off between Shapley-style fairness and Sybil resistance.

%% file: sections/02-related-work.tex
\section{Related work}\label{sec:related}

\paragraph{Data valuation and semivalues.}
Data Shapley \parencite{Ghorbani2019-vs} treats a datum's value as its expected marginal contribution to model performance over random training subsets; subsequent work has made the estimator more efficient \parencite{Jia2019-ku}, more stable under noisy utilities \parencite{Kwon2022-yc,Wang2023-cl}, and connected it to other cooperative game allocations, such as the core \parencite{Yan2021-zs}; the survey by \textcite{Rozemberczki2022-bn} places these methods alongside SHAP-style explanations \parencite{Lundberg2017-kj}. All assume that the data units are fixed before valuation begins. Volume-based valuation is replication-robust by construction \parencite{Xu2021-ul}, but does not address pure splits or paraphrase attacks. Our quotient construction is semivalue-agnostic and treats unit construction as part of the mechanism.

\paragraph{False-name-proofness and shell-company attacks.}False-name manipulation is a classic failure mode in combinatorial auctions \parencite{Yokoo2004-do,Conitzer2010-is}, and \textcite{Douceur2002-ok} argued identity multiplication cannot be prevented in distributed systems without a trusted authority or scarce resource. These insights apply to data markets, where submission value is non-rival, non-additive, and mediated by a learning algorithm. The closest work is the Faithful Group Shapley Value of \textcite{Lee2025-tt}, extending Shapley to provider groups and identifying \emph{shell-company attacks}, where a provider inflates group valuation by splitting data across auxiliary identities. The FGSV's faithfulness axiom defends against re-partitioning a fixed dataset; we study \emph{the construction of the attribution unit itself} when the manipulator can also alter data content — duplicating, paraphrasing, or synthesizing units. Our manipulation class covers near-duplicate and synthetic-variant attacks outside the FGSV data-fixed setting. The two are complementary: a deployed marketplace can use our cluster-construction layer to defeat content-level multiplication and FGSV-style axioms to defeat shell-company aggregation; we return to combinations in Sec.~\ref{sec:discussion}. \new{A related line is the Priority-Aware Shapley Value \parencite{Lee2026-jf}, which uses a supplied precedence structure together with soft priority weights and is therefore appropriate when lineage is reliable and originals and derivatives should retain separate attribution. Quotienting instead treats copies or sufficiently close variants as one strategic unit and is applicable when the direction of derivation is unavailable or untrusted. With trusted lineage, a possible combination is to form evidence-backed quotient units and apply PASV over those units using an inter-cluster precedence DAG. PASV does not itself validate a strategically reported DAG or its weights. The submitted experiments contain no PASV comparison, so we make no empirical superiority claim.}

\paragraph{Strategic data valuation and incentive-compatible pricing.}
\textcite{Zheng2025-rt} introduce data \emph{overvaluation} attacks in which clients exaggerate apparent value, showing that linear-style metrics (Shapley, LOO, Beta-Shapley, Banzhaf) are manipulable on that axis. Their concern is value misreporting at fixed identity; ours is identity manipulation at fixed value-reporting --- complementary attack surfaces. \textcite{Chen2026-qt} study truthful submission to a marketplace for mean estimation, exemplifying the IC-pricing line that distinguishes attribution fairness from cost-truthful pricing. Algorithmic data markets connect sellers and ML buyers \parencite{Agarwal2019-jb}, and data creates externalities that push market prices away from social value \parencite{Acemoglu2022-ck}. We focus on a narrow but practically important component: how to attribute and pay for training data when the submitted identity structure is manipulable.

%% file: sections/03-strategic-model.tex
\section{Strategic data-provider model}\label{sec:model}

Let $N\!=\![n]$ be latent providers. Each provider owns a multiset $D_i$ of \emph{units} drawn from a unit space $\mathcal{U}$; each unit $u$ carries a payload $\pi(u)\!\in\!\mathcal{X}\!\times\!\mathcal{Y}$ and optional metadata (id, source, timestamp). The mechanism observes a submitted profile $S\!=\!(S_1,\ldots,S_m)$ where each $S_j$ is a multiset of submitted units. Honest reporting means $m\!=\!n$ and $S_i\!=\!D_i$ as multisets; false-name reporting allows one latent provider to induce several identities, replicate units, and submit transformed copies. The learner $A$ consumes a multiset of payloads. The data-value game $([m],v_S)$ has $v_S(T)\!=\!U(A(\biguplus_{j\in T}\pi(S_j)))\!-\!U(A(\emptyset))$ where $\biguplus$ is the multiset union, $\pi(S_j)$ is the multiset of payloads in $S_j$ (duplicates preserved), $v_S(\emptyset)\!=\!0$, and (unless stated otherwise) $v_S(T)\!\in\![0,V]$. We use multisets rather than sets because replication attacks change $A$'s training input even when the unit-id set is unchanged. Monotone games suffice for the basic impossibility.

\begin{definition}[Attribution mechanism]\label{def:attribution-mechanism}
    An attribution mechanism $M$ maps every $([m],v_S)$ to payments $M(S)=(p_1,\ldots,p_m)\in\R^m$. It is \emph{budget-balanced} if $\sum_j p_j=v_S([m])$ and \emph{individually rational} if $p_j\!\geq\!0$ for every identity with non-negative contribution.
\end{definition}

$M$ is \emph{Shapley-fair} on a class of reported games if $p_i(S)\!=\!\Sh_i([m],v_S)$ for every game in that class --- exact fairness on the \emph{reported} game, which need not be the right strategic object.

We focus on two manipulation classes that directly attack attribution units. A \emph{false-name split} (Definition~\ref{def:false-name-split}) replaces honest identity $i$ with $k\!\geq\!2$ submitted identities holding datasets that union to $D_i$ (\emph{pure partition} if disjoint, \emph{replication split} otherwise). A \emph{variant attack} (Definition~\ref{def:variant-attack}) submits transformations $T_\ell(D_i)$ --- copies, augmentations, paraphrases, back-translations, refactors, synthetic generations, and benchmark-aware curations --- under one or more identities.

\begin{definition}[False-name split]\label{def:false-name-split}
    Replace identity $i$ with $k\!\geq\!2$ submitted identities $i^1,\ldots,i^k$ holding multisets $S_{i^\ell}$ with $\biguplus_\ell S_{i^\ell}\!=\!D_i$. The split is a \emph{pure partition} when the supports of the $S_{i^\ell}$ are disjoint and a \emph{replication split} when the same unit (or its payload) appears under multiple pseudonyms.
\end{definition}

\begin{definition}[Variant attack]\label{def:variant-attack}
    Submit transformed multisets $T_\ell(D_i)$ under one or more identities. Each $T_\ell$ may replicate units (preserving payload), produce variants (transformed payloads), or both.
\end{definition}

Let $S^{i\to\alpha}$ denote the profile after manipulation $\alpha$ creates the identity set $I_i^\alpha$. The \emph{additive false-name gain} is $\Gamma_i(M,S,\alpha)\!=\!\sum_{j\in I_i^\alpha}p_j(S^{i\to\alpha})\!-\!p_i(S)$; the \emph{multiplicative gain} (when $p_i(S)\!>\!0$) is $G_i\!=\!\sum_j p_j(S^{i\to\alpha})/p_i(S)$.

\begin{definition}[False-name-proofness]\label{def:false-name-proofness}
    $M$ is false-name-proof against $\M$ if $\Gamma_i(M,S,\alpha)\!\leq\!0$ for every honest $S$, latent $i$, $\alpha\!\in\!\M$; $\epsilon$-false-name-proof if the bound is $\epsilon$.
\end{definition}

Two regimes inside $\M$ matter. \emph{Replication regimes} (variants, replication splits) leave shared examples or near-duplicates across identities, recoverable by example-level evidence (perceptual hashes, embedding similarity, paraphrase detectors). \emph{Pure-split regimes} partition $D_i$ disjointly, so no example-level signal links the pseudonyms; defending against pure splits needs either a latent-ownership oracle (Theorem~\ref{thm:exact}) or non-example-level evidence (account linkage, behavioral fingerprints, license metadata, KYC).

%% file: sections/04-shapley-limits.tex
\section{Limits of Shapley-fair attribution}\label{sec:limits}

Shapley values are attractive because they are efficient, symmetric, null-player respecting, and additive in a fixed cooperative game. False-name manipulation exploits the phrase ``fixed cooperative game.'' If one provider changes the player set by splitting into several identities, exact Shapley fairness on the new reported game can conflict with false-name resistance at the latent-provider level.

\begin{example}[A complementary split increases Shapley payment]\label{ex:split}
    Consider two latent providers, $A$ and $B$, with a monotone value function
    \[
        v(\emptyset)=v(\{A\})=v(\{B\})=0,\qquad v(\{A,B\})=1.
    \]
    This is a two-player unanimity game, so $\Sh_A(v)=\Sh_B(v)=1/2$. Now suppose provider $A$ splits its dataset into two pseudonymous identities, $A_1$ and $A_2$, such that value is created only when $A_1$, $A_2$, and $B$ are all present:
    \[
        v'(T)=\ind\{\{A_1,A_2,B\}\subseteq T\}.
    \]
    This is a three-player unanimity game. Each submitted identity has Shapley value $1/3$. The latent provider $A$ now receives
    \[
        \Sh_{A_1}(v')+\Sh_{A_2}(v')=2/3>1/2=\Sh_A(v).
    \]
    Thus exact Shapley attribution over submitted identities is not false-name-proof.
\end{example}

Complementarity is natural in ML: one provider may own two rare classes, two halves of a translation pair, or two slices valuable only jointly. Splitting them creates more pivotal identities and increases the total Shapley payment.

\begin{theorem}[Exact Shapley fairness is incompatible with unrestricted false-name-proofness]
\label{thm:impossibility}
    In the class of all finite monotone data-value games, no attribution mechanism can satisfy both of the following properties:
    \begin{enumerate}[leftmargin=*,itemsep=2pt]
        \item \emph{Exact reported-game Shapley fairness:} for every submitted game $([m],v_S)$ and every submitted identity $j\in[m]$, $p_j(S)=\Sh_j(v_S)$.
        \item \emph{Unrestricted false-name-proofness:} no latent provider can increase total payment by replacing one honest identity with finitely many submitted identities.
    \end{enumerate}
    The conflict holds even when all value functions are monotonic and take values in $\{0,1\}$.
\end{theorem}

\begin{proof}[Proof sketch]
    Assume that a mechanism satisfies the exact reported-game Shapley fairness. Apply it to the honest unanimity game in Example~\ref{ex:split}; provider $A$ receives $1/2$. Apply the same mechanism to the split reported game; exact reported-game Shapley fairness requires payments $1/3,1/3,1/3$ to $A_1,A_2,B$. Because the two pseudonyms are controlled by the same latent provider, $A$ receives $2/3$ after splitting. This violates false-name-proofness. The games are monotone and binary-valued; therefore, the impossibility holds for this restricted class. A full proof formalizes the split operation and appears in App.~\ref{app:proofs}.
\end{proof}

The theorem is a boundary result, not a rejection of Shapley: Shapley fairness over raw submitted identities is the wrong primitive when identities are endogenous. The result is over identity manipulation on a \emph{fixed} game $v$; manipulating the utility function itself (re-weighting, altering the objective, changing coalition structure beyond identity multiplication) is out of scope and may relax the impossibility in either direction. Possibility returns with additional structure --- trusted identity, auditable provenance, similarity signals, or restrictions on transformations --- which the next section uses to define a quotient game with clusters as players.

The vulnerability extends beyond Shapley.

\begin{proposition}[Split-gain in unanimity games for a general semivalue]\label{prop:split-gain}
    Consider the two-player unanimity game $v$ on $\{A,B\}$ ($v(\{A,B\})\!=\!1$, zero on proper subsets). Let provider $A$ split into $k\!\geq\!2$ pseudonyms jointly holding $A$'s share, producing the $(k+1)$-player unanimity game $v'$. For any semivalue with weights $\omega$,\new{ the additive split-gain is}
    \[
        \Gamma(\omega,k) \;=\; k\,\omega_{k+1,k} - \omega_{2,1};
    \]
    multiplicatively, $G(\omega,k)\!=\!k\,\omega_{k+1,k}/\omega_{2,1}$ when $\omega_{2,1}\!>\!0$ (we report $G$ throughout the experiments and use $\Gamma$ only in the closed-form derivations). \new{A semivalue is split-vulnerable on this instance exactly when $\omega_{k+1,k}>\omega_{2,1}/k$.} The general $n$-player extension --- one of $n$ unanimity players splits into $k$ pseudonyms, producing the $(n\!+\!k\!-\!1)$-player unanimity game --- gives Shapley multiplicative split-gain $G_{\mathrm{Sh}}(n,k)\!=\!nk/(n\!+\!k\!-\!1)$, strictly above $1$ for $n,k\!\geq\!2$ and used for the closed-form predictions in App.~\ref{app:exp-extra} (S1). Proof in App.~\ref{app:proofs}.
\end{proposition}

\begin{corollary}[Specializations]\label{cor:split-gain-special}
\new{Shapley gives $\Gamma=(k-1)/(2(k+1))>0$ for all $k\geq2$, converging to $1/2$: exact Shapley fairness is split-vulnerable at every $k$. Raw Banzhaf gives $\Gamma=k\,2^{-k}-1/2$: zero at $k=2$ and \emph{negative} for $k\geq3$. Beta-Shapley interpolates continuously in its shape parameters.}
\end{corollary}

\begin{remark}[The raw-Banzhaf boundary point]\label{rem:banzhaf-boundary}
\new{The negative raw-Banzhaf \emph{additive split-gain} at $k\geq3$ is the boundary case of the characterization, not a claim that raw Banzhaf scores are negative: in this monotone game the scores are non-negative, but their total falls after the split. Raw Banzhaf therefore disincentivizes pure splitting here because its weights fail $\omega_{k+1,k}>\omega_{2,1}/k$. The impossibility of Theorem~\ref{thm:impossibility} concerns exact \emph{Shapley} fairness. Normalization can change the split-gain sign, so raw and normalized Banzhaf must be labeled separately; the benchmark columns reported here are raw unless explicitly stated otherwise.}
\end{remark}

Proposition~\ref{prop:split-gain} sharpens Theorem~\ref{thm:impossibility}: \removed{the counterexample is not Shapley-specific}\new{the split-gain formula applies to a general semivalue, while its sign depends on the weights}, and the formula gives a clean diagnostic for the experiments (Sec.~\ref{sec:experiments}, S1).\removed{ The benchmark uses normalized Banzhaf; raw Banzhaf appears only as a sanity-check column.}

\begin{remark}[Manipulation costs]\label{rem:costs}
\new{The split-gain formulas assume costless manipulation. With a per-identifier cost $c_{\mathrm{id}}$, a split into $k$ pseudonyms is profitable only while $\Gamma(\omega,k) > c_{\mathrm{id}}(k-1)$. Under Shapley, the marginal gross benefit of the $(k+1)$st pseudonym is $(n-1)/[(n+k-1)(n+k)]$, yielding a finite stopping rule under any positive linear identity cost. For an attack family $\alpha(d)$ indexed by perturbation distance $d$ and incurring cost $c(d)$, the present theory yields only the net-gain upper bound $B_i^{\alpha(d)}-c(d)=A_i^{\alpha(d)}+L_i^{\alpha(d)}+D_i^{\alpha(d)}+2K_{\alpha(d)}\eta-c(d)$. Solving the attacker's best response requires an additional model linking perturbation distance to allocation leakage, escape probability, matched-cluster semivalue drift, and retained value; that equilibrium analysis is outside the present paper.}
\end{remark}

%% file: sections/05-quotient-mechanisms.tex
\section{Quotient semivalue mechanisms}\label{sec:quotient}

\subsection{Evidence graph and attribution clusters}

Let $\U\!=\!\{u_1,\ldots,u_r\}$ be submitted data units (granularity --- example, document, batch, file --- is a designer choice), each with submitted identity $o(u)\!\in\![m]$. The mechanism receives an evidence function $e_\theta(u,u')\!\in\!\{0,1\}$ indicating whether two units should be treated as the same attribution object at threshold $\theta$ (combining hashes, perceptual hashes, embedding similarity, license metadata, watermarks, timestamps, or trusted provenance --- the theory conditions on error properties, not on the technology). Build the graph $G_\theta\!=\!(\U,E_\theta)$ with edges where $e_\theta\!=\!1$, and let $\C_\theta\!=\!\{C_1,\ldots,C_K\}$ be its connected components --- the mechanism's \emph{attribution clusters}.

Clustering identities are necessary but not sufficient: A manipulator who has placed ten copies of one example into a cluster can still shift the trained model if the learner re-weights that point ten times. The strategic unit is the identity; the operational unit consumed by the learner is the training set. We close the gap with a \emph{cluster representative} operator that canonicalizes each cluster's training units before the utility is evaluated.

\begin{definition}[Cluster representative]\label{def:representative}
    A cluster representative is a map $R_\theta:\,2^{\U}\to 2^{\U^{\star}}$ from clusters to (multisets of) canonical training units. Concrete instantiations include:
    \begin{enumerate}[leftmargin=*,itemsep=2pt]
        \item \emph{Exact-duplicate collapse:} replace identical units in $C_k$ by a single representative.
        \item \emph{Capped representative:} retain a fixed budget of $\kappa$ units per cluster (e.g., the centroid, medoid, or first $\kappa$ in submission order), discarding additional copies.
        \item \emph{Weight-normalized union:} use $\bigcup_{u\in C_k} u$ but rescale per-example training weight to a cluster-level budget that is independent of $|C_k|$.
        \item \emph{Provenance-based selection:} retain only units whose source ID is recorded as the cluster's canonical origin.
    \end{enumerate}
\end{definition}

\begin{definition}[Quotient game]\label{def:quotient-game}
    Given submitted units $\U$, clusters $\C_\theta=\{C_1,\ldots,C_K\}$, cluster representative $R_\theta$, learner $A$, and utility $U$, the quotient data-value game is $([K],\bar v_\theta)$, where
    \[
        \bar v_\theta(Q)=U\!\left(A\!\left(\bigcup_{k\in Q} R_\theta(C_k)\right)\right)-U\!\left(A(\emptyset)\right)
    \]
    for every coalition of clusters $Q\subseteq[K]$.
\end{definition}

The quotient operation changes the strategic unit from a submitted identity to an evidence-backed cluster. It yields invariance only for manipulations that leave the evidence graph's cluster structure and the representative outputs unchanged; Definitions~\ref{def:quotient-stable} and~\ref{thm:approx} below make these conditions and their failure terms explicit.

\paragraph{Payment evaluation and deployment.}\new{In the experiments, $R_\theta(C_k)$ is used in the coalition utility calls that determine payments; we do not separately evaluate a final deployed model. The demonstrated guarantees are therefore payment guarantees. If a final model is trained on raw submissions, duplicate reweighting can remain even when payments are robust; removing that channel requires training on the same canonicalized or weight-normalized representation. Exact-duplicate collapse is appropriate for literal replication; weight-normalized union retains all content while fixing total cluster influence; capped or medoid representatives can change coalition utilities and thereby contribute to the matched-cluster semivalue-drift term $D_i^\alpha$; provenance-based selection is preferable when a trusted origin is available.}

\begin{remark}[Two evidence regimes]\label{rem:regimes}
$e_\theta$ is example-level, as it captures replication regimes (duplicates, near-duplicates, and paraphrases) but is silent on pure splits, in which no two units across identities are particularly close. Defending pure splits requires non-example-level evidence (account linkage, behavioral fingerprints, and trusted provenance); the theorems below are conditional on whichever evidence layer the designer adopts.
\end{remark}

\subsection{Semivalue payments and within-cluster allocation}

A semivalue assigns to cluster $k$ in the quotient game $\bar v_\theta$ a weighted average of marginal contributions $\varphi_k^\omega(\bar v_\theta)\!=\!\sum_{Q\subseteq[K]\setminus\{k\}}\omega_{K,|Q|}[\bar v_\theta(Q\cup\{k\})-\bar v_\theta(Q)]$, with weights $\omega_{K,s}\!\geq\!0$ satisfying $\sum_s\binom{K-1}{s}\omega_{K,s}\!=\!1$. Shapley weights are $1/(K\binom{K-1}{s})$, Banzhaf weights $2^{-(K-1)}$, Beta-Shapley uses a parametric coalition-size tilt \parencite{Kwon2022-yc}. Shapley weights are efficient ($\sum_k\!\varphi_k^\omega\!=\!\bar v_\theta([K])$, so payments are budget-balanced); Banzhaf and Beta are not. Rescaling their cluster scores to the grand-coalition value creates normalized payments, but this is a separate operation with extra stability requirements (Lemma~\ref{lem:normalized}). The benchmark's Banzhaf columns are raw unless explicitly labeled; raw and normalized split-gain signs can differ (Proposition~\ref{prop:split-gain}).

\begin{definition}[Quotient semivalue attribution]\label{def:qsv}
    Let $a_{i,k}\!\in\![0,1]$ be the share of cluster $C_k$ allocated to submitted identity $i$, with $\sum_i a_{i,k}\!=\!1$. The quotient semivalue payment is $p_i^{Q,\omega}(S)\!=\!\sum_k a_{i,k}\,\varphi_k^\omega(\bar v_\theta)$.
\end{definition}

\begin{assumption}[False-name-neutral within-cluster allocation]\label{ass:false-name-neutral}
    Let latent provider $i$ replace identity $i$ with pseudonyms $I_i^\alpha$, with pseudonym units possibly distributed across multiple clusters. Then for every cluster $k\in[K]$ the share sum is preserved:
    \[
        \sum_{j\in I_i^\alpha}a_{j,k}^\alpha=a_{i,k}\qquad\text{for every }k.
    \]
    The equality is required cluster by cluster for one common manipulated profile and allocation; it is not enough that each cluster could satisfy the equality under a different hypothetical allocation.
\end{assumption}

\paragraph{Within-cluster allocation is a second design axis.}
Assumption~\ref{ass:false-name-neutral} is sufficient for the zero-allocation-leakage special case but does not pin down the within-cluster rule. Three natural rules satisfy it under different conditions: \emph{equal-share} ($a_{j,k}\!=\!1/|I_k|$) is neutral only in single-latent clusters and fails in mixed clusters where a split into $h$ pseudonyms among $g$ other IDs raises the latent share from $1/(g+1)$ to $h/(g+h)$; \emph{count-based} over canonical units ($\propto |R_\theta(C_k)\cap S_j|$) is neutral only when canonical ownership counts are well defined and invariant to the attack, which requires a fixed rule for assigning a collapsed duplicate to an owner; and \emph{latent-share} is neutral when reliable ownership provenance is available. App.~\ref{app:allocation-rules} (Table~\ref{tab:allocation-rules}) gives the precise conditions and failure modes. Empirically (S5/S7), the three rules diverge by an order of magnitude under false merges: the equal-share gain rises to $2.46$ at $p_{\mathrm{fm}}\!=\!0.40$, and count-based allocation halves the false-merge effect in this benchmark. Count-based allocation is therefore our preferred baseline when canonical ownership counts are well defined and attack-invariant, not a universal deployment default. We use permutation/random-subset sampling on $K\!\ll\!r$ clusters (App.~\ref{app:algorithms}).

\paragraph{The role of Assumption~\ref{ass:false-name-neutral}.}\new{Assumption~\ref{ass:false-name-neutral} is a design condition on the within-cluster rule, not a restatement of the conclusion. The quotient construction makes cluster-level invariance \emph{testable}: a manipulation leaves semivalues unchanged when the evidence graph and representative operator leave the quotient game unchanged. Quotienting alone does not guarantee that condition; escaped clusters and matched-cluster semivalue drift measure its failure in Theorem~\ref{thm:approx}. Even when cluster values are stable, the allocation layer can leak share under identity multiplication. Assumption~\ref{ass:false-name-neutral} sets that leakage to zero, while the general theorem retains it explicitly as $A_i^\alpha$. No analogous allocation repair exists for semivalues computed over raw identities, where splitting changes the player set and game before allocation (Theorem~\ref{thm:impossibility}, Proposition~\ref{prop:split-gain}).}

\paragraph{Scaling when quotienting does not compress.}\new{When most honest units are unique, $K\approx r$, quotient valuation can be substantially more expensive than Shapley over the $m$ submitted providers, and clustering adds overhead. For permutation Shapley, $R=O(V^2\log(K/\delta)/\eta^2)$ permutations give a uniform per-cluster error target up to constants, with $O(RK)$ coalition-prefix evaluations. Random-subset sampling is the corresponding estimator for raw Banzhaf, followed by normalization when scores are used as payments. Our experiments operate where evidence aggregation keeps $K$ small; they do not establish scalability at $K\approx r$. Coarser attribution units, hierarchical valuation, or utility surrogates would be needed in that regime.}

%% file: sections/06-guarantees.tex
\section{Guarantees}\label{sec:guarantees}

A two-tier hierarchy. Theorem~\ref{thm:approx} bounds manipulation gain through four distinct channels: allocation leakage $A_i^\alpha$, escaped-cluster mass $L_i^\alpha$, matched-cluster semivalue drift $D_i^\alpha$, and estimator error $\eta$. Theorem~\ref{thm:fairness} separately bounds cluster-level fairness loss by the quotient-game distance $\Delta_\theta$. Cosine and noisy-evidence mechanisms can be evaluated under the general bound when the stated cluster matching exists; rules satisfying Assumption~\ref{ass:false-name-neutral} set only $A_i^\alpha$ to zero. Theorem~\ref{thm:exact} is the zero-leakage case.

\begin{definition}[Quotient-stable manipulation]\label{def:quotient-stable}
    Manipulation $\alpha$ by latent provider $i$ is \emph{quotient-stable} (with respect to evidence $e_\theta$ and representative $R_\theta$) if it leaves both (a)~the set of quotient clusters $\C_\theta$ unchanged up to relabeling of cluster IDs, and (b)~each cluster representative $R_\theta(C_k)$ unchanged as a multiset of canonical training units. Pure splits and replication attacks, whose variants all fall into existing honest clusters and are absorbed by $R_\theta$ are quotient-stable; attacks that produce variants escaping into new clusters or that change a cluster's canonical representative are not.
\end{definition}

\begin{theorem}[Exact false-name-proofness under quotient-stability]\label{thm:exact}
    Suppose that manipulation $\alpha$ by latent provider $i$ is quotient-stable (Definition~\ref{def:quotient-stable}) and the within-cluster allocation rule is false-name-neutral. Then, for any semivalue weights $\omega$ and any value function $\bar v_\theta$, exact quotient semivalue attribution satisfies $\Gamma_i(M^{Q,\omega},S,\alpha)=0$, that is, the mechanism is exactly false-name-proof against $\alpha$.
\end{theorem}

The theorem isolates the limit-case conditions: the evidence graph and representative operator keep the quotient game unchanged, while the allocator preserves the latent provider's aggregate share (App.~\ref{app:proofs}). A latent-truth oracle satisfies the first condition. Source ID evidence does so for attacks whose copies or variants retain a common source identifier; a disjoint pure split needs provider-, account-, or license-level linkage. Cosine evidence need not be quotient-stable because variants may escape or alter representatives.

To compare the honest and manipulated quotient games, let $H_i=\{h\in[K]:a_{i,h}>0\}$ be the honest clusters that pay provider $i$. An \emph{admissible matching} is a bijection $\mu:M_i^\alpha\to H_i$ from selected attacker-containing clusters in the manipulated game to these honest clusters. Every other attacker-containing manipulated cluster is escaped and belongs to $E_i^\alpha$. Thus an honest payment term is used exactly once: when several manipulated clusters descend from one honest cluster, one may be matched and the others are escaped; a many-to-one correspondence is not silently reused. Manipulations that merge several of provider $i$'s honest clusters into one attacked cluster do not admit this matching and fall outside Theorem~\ref{thm:approx}.

For $k\in M_i^\alpha$, write $s_{i,k}^\alpha=\sum_{j\in I_i^\alpha}a_{j,k}^\alpha$ for the attacker's aggregate share. Define allocation leakage, escaped mass, and matched-cluster semivalue drift by
\[
\begin{aligned}
    A_i^\alpha&=\sum_{k\in M_i^\alpha}\left|s_{i,k}^\alpha-a_{i,\mu(k)}\right|\left|\varphi_k^\omega(\bar v_\theta^\alpha)\right|,\\
    L_i^\alpha&=\sum_{k\in E_i^\alpha}\left|\varphi_k^\omega(\bar v_\theta^{\alpha})\right|,\\
    D_i^\alpha&=\sum_{k\in M_i^\alpha}\left|\varphi_k^\omega(\bar v_\theta^{\alpha})-\varphi_{\mu(k)}^\omega(\bar v_\theta)\right|.
\end{aligned}
\]
$A_i^\alpha$ measures within-cluster share creation, $L_i^\alpha$ measures value presented as new attribution units, and $D_i^\alpha$ measures semivalue drift in the matched comparison.

\begin{theorem}[Approximate false-name-proofness under evidence and allocation leakage]\label{thm:approx}
    Suppose an admissible matching $\mu$ exists. Let the honest and manipulated mechanisms both use semivalue estimates, with $|\widehat\varphi_h-\varphi_h|\!\leq\!\eta$ for every honest cluster and $|\widehat\varphi_k^\alpha-\varphi_k^\alpha|\!\leq\!\eta$ for every manipulated cluster. Then
    \[
        \left|\sum_{j\in I_i^\alpha}\widehat p_j^{Q,\omega}(S^\alpha)-\widehat p_i^{Q,\omega}(S)\right|\;\leq\; B_i^\alpha:=A_i^\alpha+L_i^\alpha+D_i^\alpha+2K_\alpha\eta,
    \]
    and hence $\Gamma_i(\widehat M^{Q,\omega},S,\alpha)\leq B_i^\alpha$, where $K_\alpha$ is the manipulated cluster count. If instead both sets of estimates satisfy $\E|\widehat\varphi-\varphi|\leq\eta$, then $\E|\sum_j\widehat p_j^\alpha-\widehat p_i|\leq B_i^\alpha$ and therefore $\E[\Gamma_i]\leq B_i^\alpha$. For permutation Shapley, when each sampled marginal contribution has range width $O(V)$, $O(V^2\log((K+K_\alpha)/\delta)/\eta^2)$ samples per cluster give a uniform high-probability error event, while the mean-absolute-error regime needs $O(V^2/\eta^2)$. In particular, $A_i^\alpha+L_i^\alpha+D_i^\alpha\leq\epsilon$ and $\eta\leq\epsilon/(2K_\alpha)$ give $2\epsilon$-false-name-proofness.
\end{theorem}

Assumption~\ref{ass:false-name-neutral} on every matched pair gives $A_i^\alpha=0$. Together with quotient stability, it also gives $L_i^\alpha=D_i^\alpha=0$; exact evaluation sets $\eta=0$, recovering Theorem~\ref{thm:exact}. Equal-share in false-merged clusters need not have $A_i^\alpha=0$, so those experiments are interpreted through the general bound rather than the zero-leakage special case.

\begin{remark}[When matched-cluster semivalue drift vanishes]\label{rem:drift-zero}
$D_i^\alpha\!=\!0$ when the entire honest and manipulated quotient games are equal up to the relabeling induced by $\mu$, as in the quotient-stable case. Equality of the representatives for matched clusters removes their direct local-representative change but is not sufficient: escaped clusters or changes elsewhere in the quotient game can still change a matched cluster's semivalue. Quantitative bounds based on medoid displacement or representative-weight perturbation therefore require an additional assumption translating those local changes into a uniform bound on coalition-value changes.
\end{remark}

A probabilistic-miss corollary (proof in App.~\ref{app:proofs}): if each manipulated variant escapes its honest cluster with probability at most $\delta$ and cluster marginal contributions are bounded by $B$, then $\E[L_i^\alpha]\leq \delta r_i B$. Bounding $D_i^\alpha$ analogously requires a global stability condition connecting each variant to the induced changes in coalition values and hence in matched-cluster semivalues.

\begin{theorem}[Fairness loss at the cluster level]\label{thm:fairness}
    Let $v$ be an honest data-value game over honest attribution units and $\bar v_\theta$ its quotient game under clustering map $q:\,\U\to[K]$ and representative $R_\theta$. For each cluster $k$, let $C^q_k=q^{-1}(k)$ denote the set of honest units mapped to cluster $k$, and define the merged honest game on player set $[K]$ by $v^{\mathrm{merge}}(Q):=v(\bigcup_{\ell\in Q}C^q_\ell)$. Suppose for all coalitions $Q\subseteq[K]$,
    \[
        \left|v^{\mathrm{merge}}(Q) - \bar v_\theta(Q)\right|\leq \Delta_\theta.
    \]
    Then for any semivalue weights $\omega$, the cluster-level fairness loss is bounded:
    \[
        \left|\varphi_k^\omega(v^{\mathrm{merge}}) - \varphi_k^\omega(\bar v_\theta)\right|\leq 2\Delta_\theta
        \quad\text{for every cluster } k.
    \]
\end{theorem}

Theorem~\ref{thm:fairness} bounds cluster-level fairness loss; considered alongside Theorem~\ref{thm:approx}, it separates the fairness and robustness quantities that the empirical threshold analysis trades off. Individual-level fairness requires a within-cluster allocation rule. Theorems~\ref{thm:exact}--\ref{thm:fairness} concern raw semivalues. Lemma~\ref{lem:normalized} extends the manipulation bound to normalized payments only under explicit lower bounds on both raw-score sums and independent bounds on the raw-sum and grand-coalition-value changes. Without those stability conditions, normalization has no inherited guarantee. Finally, observationally identical reports do not reveal whether two accounts have common ownership; provenance is therefore needed to justify an ownership-sensitive allocation, but replication resistance and symmetry alone are not logically incompatible (App.~\ref{app:proofs}).

%% file: sections/07-experiments.tex
\section{Experiments}\label{sec:experiments}

We evaluate quotient semivalues on \textsc{DataMarket-Gym}, a benchmark for training-data markets with strategic providers (full specs in App.~\ref{app:gym}): a provider generator, an attack library, a controllable evidence layer, multiple valuation backends ($k$-NN, logreg, frozen ResNet/transformer features), and standard metrics (utility, manipulation gain, oracle-$L^1$ loss, rank stability, runtime). The experiments answer three questions: (\textbf{Q1}, Sec.~\ref{sec:exp-synthetic}) do quotient semivalues recover Theorem~\ref{thm:exact}'s predicted false-name-proofness on synthetic ground-truth, and do Theorems~\ref{thm:approx}--\ref{thm:fairness}'s bounds track empirical loss as evidence quality degrades; (\textbf{Q2}, Sec.~\ref{sec:exp-real}) does the mechanism transfer to real ML pipelines with frozen features and no oracle access; (\textbf{Q3}, Sec.~\ref{sec:exp-real}) can a usable cosine-threshold interval be diagnosed before deployment from embedding-pool geometry under a prespecified variant family, provider-partition model, graph construction, and mixed-component tolerance.

\subsection{Synthetic experiments}\label{sec:exp-synthetic}

The synthetic pipeline runs cooperative-game benchmarks (where exact theoretical predictions are available) and a logistic regression learner on a synthetic classification task ($8$ providers, $60$ examples each, $4$ classes, $24$ features, class separation $1.2$). In unanimity games, the closed-form Shapley and raw Banzhaf split-gain predictions of Proposition~\ref{prop:split-gain} match measurements to three decimal places at every $k\!\in\!\{2,\ldots,6\}$ (App.~\ref{app:exp-extra}, Table~\ref{tab:split-gain-validation}); the proposition is numerically validated. DGP robustness over $n_{\mathrm{providers}}\!\in\!\{6,8\}$ and class-separation $\in\!\{0.6,1.2,1.8\}$ (App.~\ref{app:exp-extra}, S8) confirms the synthetic findings below are not artifacts of a single configuration.

\paragraph{S2: Main attack--mechanism table.}
For the synthetic classification task, we compared four attribution mechanisms (baseline Shapley over submitted identities and quotient Shapley with three evidence layers: latent oracle, source ID provenance, and cosine $\theta\!=\!0.99$) across five attacks (honest baseline, pure Sybil split with $k\!=\!3$, exact-duplicate-with-Sybil-split, near-duplicate-with-Sybil-split, and a label-noise poisoning attack). All numbers are the mean manipulation gain $\pm$ SE over $50$ seeds.

\begin{table}[tbp]
\centering
\small
\caption{Manipulation gain $G_i$ on the synthetic classification task ($50$ seeds, mean $\pm$ SE). Lower is better; $G_i\!=\!1$ means the attack does not change total payment. Bold entries highlight where baseline Shapley fails. Bottom panel: quotient mechanisms over evidence-backed clusters, where bold entries indicate near-honest gain. \emph{Uniform} (each submitted ID gets $1/m$), \emph{Per-ex.}\ (each example gets $1/r$), \emph{LOO}, \emph{Shap}, \emph{Banz}\ (raw submitted-identity Banzhaf, not normalized), \emph{$\beta$-Shap}\ (Beta-Shapley, $\alpha\!=\!\beta\!=\!2$). Quotient: \emph{Sh.\,oracle} (latent-truth oracle), \emph{Bz.\,oracle} (raw quotient Banzhaf with latent oracle), \emph{Sh.\,source} (source ID provenance), \emph{Sh.\,cos.$0.99$} (cosine $\theta\!=\!0.99$). No normalized-Banzhaf column is reported.}\label{tab:main-table}

\textbf{Submitted-identity baselines.}\\[0.4ex]
\begin{tabular}{lcccccc}
\toprule
Attack & Uniform & Per-ex.\ & LOO & Shap & Banz & $\beta$-Shap \\
\midrule
honest & $1.000$ & $1.000$ & $1.000$ & $1.000$ & $1.000$ & $1.000$ \\
sybil $k\!=\!3$ & $2.40$ & $1.00$ & $0.41{\scriptstyle \pm 0.36}$ & $\mathbf{1.60{\scriptstyle \pm 0.02}}$ & $0.93{\scriptstyle \pm 0.02}$ & $1.41{\scriptstyle \pm 0.04}$ \\
exact dup.\,+\,sybil & $2.39$ & $1.77$ & $0.35{\scriptstyle \pm 0.46}$ & $\mathbf{1.74{\scriptstyle \pm 0.03}}$ & $0.70{\scriptstyle \pm 0.04}$ & $1.48{\scriptstyle \pm 0.05}$ \\
near dup.\,+\,sybil & $2.39$ & $1.77$ & $0.40{\scriptstyle \pm 0.47}$ & $\mathbf{1.74{\scriptstyle \pm 0.03}}$ & $0.70{\scriptstyle \pm 0.04}$ & $1.48{\scriptstyle \pm 0.05}$ \\
label noise & $0.98$ & $0.98$ & $-0.26{\scriptstyle \pm 0.41}$ & $0.57{\scriptstyle \pm 0.02}$ & $0.20{\scriptstyle \pm 0.03}$ & $0.45{\scriptstyle \pm 0.03}$ \\
\bottomrule
\end{tabular}

\vspace{0.6ex}
\textbf{Quotient mechanisms over evidence-backed clusters.}\\[0.4ex]
\begin{tabular}{lcccc}
\toprule
Attack & Sh.\,oracle & Bz.\,oracle & Sh.\,source & Sh.\,cos.$0.99$ \\
\midrule
honest & $1.000$ & $1.000$ & $1.000$ & $1.000$ \\
sybil $k\!=\!3$ & $\mathbf{1.00}$ & $\mathbf{1.00}$ & $1.61{\scriptstyle \pm 0.02}$ & $1.61{\scriptstyle \pm 0.02}$ \\
exact dup.\,+\,sybil & $\mathbf{0.96{\scriptstyle \pm 0.01}}$ & $0.87{\scriptstyle \pm 0.03}$ & $\mathbf{0.96{\scriptstyle \pm 0.01}}$ & $\mathbf{0.96{\scriptstyle \pm 0.01}}$ \\
near dup.\,+\,sybil & $\mathbf{0.96{\scriptstyle \pm 0.01}}$ & $0.87{\scriptstyle \pm 0.03}$ & $\mathbf{0.96{\scriptstyle \pm 0.01}}$ & $\mathbf{0.96{\scriptstyle \pm 0.01}}$ \\
label noise & $0.57{\scriptstyle \pm 0.02}$ & $0.20{\scriptstyle \pm 0.03}$ & $0.57{\scriptstyle \pm 0.02}$ & $0.57{\scriptstyle \pm 0.02}$ \\
\bottomrule
\end{tabular}
\end{table}

\paragraph{Reading the table.}
On replication attacks, every quotient mechanism with example-level evidence drops $G$ to the latent-oracle level $\approx\!0.96$, against $1.74$ for baseline Shapley. Pure splits are matched only by the latent oracle (Remark~\ref{rem:regimes}): source ID evidence handles copies or variants that retain a common source identifier, not disjoint source units split across pseudonyms, which explains its $G\!=\!1.61$ pure-split result. Raw Banzhaf has multiplicative gain below one ($G\!=\!0.93$ on Sybil $k\!=\!3$); raw quotient Banzhaf with latent oracle gets $G\!=\!1$ on Sybil and $0.87$ on duplicates, showing that quotienting can be paired with different semivalues. LOO is large and unstable on every Sybil attack. Split-factor $k\!\in\!\{2,3,4\}$ gives Shapley gains $1.37, 1.60, 1.78$, qualitatively matching $nk/(n+k\!-\!1)$. We use $s\!=\!256$ Shapley samples; budget details in App.~\ref{app:exp-extra}.

\paragraph{S4: Fairness--Sybil threshold frontier.}
Sweeping cosine threshold $\theta\!\in\!\{0.85,\ldots,0.99\}$ on near-duplicate Sybil ($\sigma\!=\!0.03$, $50$ seeds), manipulation gain transitions sharply from $1.246\!\pm\!0.052$ at $\theta\!=\!0.85$ to $0.964\!\pm\!0.009$ for $\theta\!\geq\!0.95$ (latent-oracle level), against baseline Shapley $1.738$; oracle $L^1$ fairness loss tracks the same frontier ($0.162\!\to\!0.015$). On \emph{pure-split} Sybil (no example overlap), no $\theta$ defends and a loose $\theta\!=\!0.85$ is actively counterproductive ($G\!=\!1.84$ against baseline $1.60$): false-merges between independent providers amplify Sybil shares, as Remark~\ref{rem:regimes} predicts. Full curve, second-axis pure-split data, and Figure~\ref{fig:frontier} in App.~\ref{app:exp-extra}.

\paragraph{Two further synthetic findings (App.~\ref{app:exp-extra}).}
S5 sweeps a $5\!\times\!5$ oracle-noise grid (false-split $p_{\mathrm{fs}}$, false-merge $p_{\mathrm{fm}}$, Table~\ref{tab:noise-grid}; both noise types are i.i.d.\ per edge --- a simplification: real clustering errors are correlated with provider size, embedding density, and label structure): under equal-share allocation, false-merges dominate ($G\!=\!2.46$ at $p_{\mathrm{fm}}\!=\!0.40$, worse than baseline) while $p_{\mathrm{fs}}$ alone only reaches $G\!=\!1.13$. The false-merge increase is an empirical allocator-leakage channel $A_i^\alpha>0$, not a regime satisfying Assumption~\ref{ass:false-name-neutral}. Fairness loss tracks $p_{\mathrm{fm}}$ almost independently of $p_{\mathrm{fs}}$, validating Theorem~\ref{thm:fairness}'s bound. S7 shows that the dominance is rule-specific: count-based or latent-share allocation more than halves the false-merge effect on near-dups Sybil and pushes pure-Sybil gain \emph{below} honest (Fig.~\ref{fig:allocation-sensitivity}). In this DGP the latter rules reduce $A_i^\alpha$ because ownership counts map cleanly to latent providers; count-based allocation is not generally neutral unless canonical counts are attack-invariant. Plotting empirical loss against the mixed-component fraction across both ablations confirms Theorem~\ref{thm:fairness}'s linear bound is conservative and saturates near $0.4$.

\subsection{Image- and text-domain experiments}\label{sec:exp-real}

Q1's synthetic game results show the mechanism behaves as theory predicts when ground-truth clusters are known. Q2 explores if the same defense holds in production-style ML pipelines, where evidence is from frozen feature extractors and the only oracle is cosine similarity over those features. We replicate the S2 mechanism set on CIFAR-10 with frozen ImageNet ResNet-18 features ($512$-d) and on AG News with frozen MiniLM-L6-v2 sentence embeddings ($384$-d), four providers $\times$ $50$ examples each, $10$ seeds, plus a cosine threshold sweep $\theta\!\in\!\{0.85,0.90,0.95,0.99\}$. In the near-duplicate Sybil attack with $\sigma\!=\!0.02$, baseline Shapley shows a manipulation gain $1.07$ on CIFAR and $1.72$ on text; quotient-Shapley with the latent oracle reduces both to $1.00$, confirming synthetic findings transfer. The full table is in App.~\ref{app:exp-extra}.

\paragraph{Threshold-reversal across domains, and its mechanism.}The cosine threshold sweep reveals a domain dependency: CIFAR's tighter $\theta$ saturates the defense ($\theta\!\geq\!0.90$: $G\!=\!1.00$), while in text the direction \emph{reverses} ($\theta\!\leq\!0.90$: $G\!=\!1.00$; $\theta\!\geq\!0.95$: $G\!=\!1.78$). The reversal is driven by which of the two lower-bound failure modes binds: \emph{pairwise over-merge} or \emph{chaining}. The upper bound, near-duplicate cosine $\cos(x,x\!+\!\eta)$ at $\sigma\!=\!0.02$ depends on feature norm — CIFAR's $\|x\|\!\approx\!28$ pushes it to $\approx\!0.9998$ (non-binding); unit-normalized MiniLM drops it to $0.92$. The pairwise lower bound is the intra-class cosine p90 (typical honest same-class pair) — $0.81$ for CIFAR (binding) and $0.26$ for MiniLM (non-binding). When pairwise is non-binding, chaining binds: sub-threshold cross-provider edges accumulate, collapsing the provider-level evidence graph into one giant component. We diagnose chaining by simulating the provider-level mixed-component fraction (MCF) on the embedding pool (matching the experimental graph protocol); the \emph{chaining floor} is the smallest $\theta$ at which the simulated MCF$<\!0.10$. The binding lower bound is the larger of the two; in all four domains, chaining binds: CIFAR $0.88$, AG~News $0.84$, IMDB $0.78$; STL-10 is a tie ($0.78$). The MCF simulation matches the experimental task partition (random for text and class-stratified for image); App.~\ref{app:exp-extra}. Pure-split Sybil on both domains is matched only by the latent oracle (Remark~\ref{rem:regimes}).

\paragraph{Calibration-domain and held-out-domain checks.}
The MCF cutoff ($0.10$) was calibrated against the original IMDB sweep at $\theta\!\in\!\{0.50, 0.70, 0.85, 0.90, 0.95\}$, so IMDB is a calibration domain rather than a held-out test. Its additional cells $\theta\!\in\!\{0.75,0.78,0.80,0.82\}$ provide retrospective sharpness checks around the calibrated floor: the predicted interval is $[0.78,0.92]$, and the added sweep reaches oracle level at $0.78$--$0.90$ while failing below and above that interval. STL-10 is the genuine held-out domain: its interval $[0.78,1.00]$ was computed before the attack-defense sweep, which gives $G\!=\!2.00$ at $\theta\!\leq\!0.65$, $1.70$ at $0.75$, and $0.999$ at $\theta\!\geq\!0.80$. The SE per cell over 10 seeds is $\leq\!0.01$ inside the admissible intervals and $\leq\!0.16$ in transition cells. These checks use embedding-pool geometry under the calibrated Gaussian-variant family, provider-partition simulation, connected-component graph, and MCF tolerance implemented by \texttt{scripts/predict\_theta.py}; they are not predictions from embedding statistics alone. Per-$\theta$ values appear in Table~\ref{tab:holdout}.

\paragraph{Pre-deployment admissibility check.}\new{The threshold-prediction machinery doubles as a deployment safeguard against the false-merge regime in which cosine-quotient underperforms baseline Shapley (S4). From an embedding pool and a prespecified variant family, form $[\theta_{\min},\theta_{\max}]$, where $\theta_{\min}$ is the larger of the pairwise-similarity p90 and the MCF chaining floor, and $\theta_{\max}$ is the near-duplicate-similarity p10. Use cosine quotienting only when this interval is nonempty with a useful margin and the simulated mixed-component fraction is below the chosen tolerance; otherwise use provenance-grade evidence or the identity-level baseline. This is an empirically calibrated diagnostic, not a theorem-level certificate. It concerns duplicate and variant evidence and does not solve pure splits without non-example-level evidence.}

%% file: sections/08-discussion.tex
\section{Discussion, limitations, and conclusion}\label{sec:discussion}

\textbf{The unit of attribution is a mechanism-design choice.} Data Shapley and its successors treat units as given; in markets, the apparent granularity is endogenous. Theorem~\ref{thm:impossibility} rules out exact Shapley + unrestricted FNP; quotient semivalues move the strategic burden to evidence-backed clusters, semivalue-agnostically. Two design levers matter: \emph{within-cluster allocation} (count-based over canonical units is preferred when canonical ownership counts are well defined and attack-invariant; equal-share is fragile under false-merges) and the \emph{cosine threshold} $\theta$ (diagnosed from pairwise p90, chaining-floor MCF, and near-duplicate p10 under the calibrated variant and graph model, Sec.~\ref{sec:exp-real}).\removed{ Complementarity with FGSV: cluster construction defeats content-level multiplication, and FGSV-style axioms defeat shell-company aggregation.}

\new{\textbf{Combining with faithfulness axioms.} For the Shapley case, a clean composition is available when quotient units have a fixed group assignment. First construct canonical quotient units using the evidence graph and representative operator, then compute each quotient unit's Shapley value $\phi_k(\bar v)$. For a fixed group $G$ of quotient units, define its FGSV payment as $P_G=\sum_{k\in G}\phi_k(\bar v)$. If $G=G_1\mathbin{\dot\cup}\cdots\mathbin{\dot\cup}G_h$, then $\sum_\ell P_{G_\ell}=P_G$, so repartitioning a fixed quotient dataset does not change the total. The quotient layer addresses changes to content units, subject to quotient stability or Theorem~\ref{thm:approx}'s leakage terms, while FGSV addresses regrouping of those fixed units. This composition is specific to Shapley/FGSV and is not automatic for arbitrary semivalues. Mixed-ownership clusters also require a valid within-cluster allocation rule; we do not claim an unrestricted composition theorem.}

\new{\textbf{Limitations.} Our scope is false-name and variant manipulation on a fixed underlying data pool. Poisoning, cost-truthfulness, and content-truthfulness are outside the mechanism's reach and need complementary instruments --- procurement design, audits, or legal enforcement; evidence from platform auditing shows that API restrictions can create audit blind spots even under formal transparency mandates \parencite{Burnat2026-xq}. The evidence graph is a two-sided design risk: weak graphs leak manipulation mass through escaped clusters ($L_i^\alpha$), while aggressive graphs falsely merge independent providers, can create allocation leakage ($A_i^\alpha$), and underpay them (S4--S7). Pure-split attacks without content overlap defeat example-level evidence entirely; source IDs cover them only when the split units retain common lineage, otherwise account-, provider-, or license-level evidence is needed (Remark~\ref{rem:regimes}). Provenance metadata strengthens evidence at a privacy cost to providers. The attack library is dual-use: the same primitives that stress-test mechanisms describe how to game deployed ones. Finally, our guarantees are conditional on the matching comparison, matched-cluster semivalue drift, allocation behavior, and estimator accuracy; each channel is explicit, but none is free. The attack library contains mechanism-aware templates rather than derived best responses: the theorems bound attacks satisfying their stated comparison conditions, while which attacks rational manipulators would choose under explicit payoffs and costs is an equilibrium layer left for future work.}

\new{\textbf{Conclusion.} Attribution units in data markets are strategic objects. We formalized false-name manipulation in ML data attribution, proved that exact reported-game Shapley fairness is incompatible with unrestricted false-name-proofness, and showed that semivalues over evidence-backed clusters restore exact false-name-proofness under quotient stability and allocator neutrality. Away from that limit, manipulation gain decomposes into allocation leakage, escaped mass, matched-cluster semivalue drift, and estimator error, while cluster-fairness loss has a separate quotient-distance bound. DataMarket-Gym shows how these channels behave across synthetic, image, and text domains. Under a calibrated variant, partition, and graph model, embedding-pool geometry identifies useful threshold ranges; STL-10 supplies the held-out test, while IMDB is a calibration-domain check. The unit of attribution, not the semivalue on top of it, is where false-name robustness is won or lost.}

%% file: sections/appendix-a-proofs.tex
\section{Proofs}\label{app:proofs}

\subsection{Proof of Theorem~\ref{thm:impossibility}}

Let $M$ be any mechanism satisfying exact reported-game Shapley fairness. Consider the honest two-player unanimity game $v$ in Example~\ref{ex:split}. Since $M$ is exactly Shapley-fair, $M_A(v)=1/2$. Now consider the split report $v'$ in which $A$ appears as two pseudonyms $A_1,A_2$ and the value is generated only by the grand coalition $\{A_1,A_2,B\}$. Exact reported-game Shapley fairness gives $M_{A_1}(v')=M_{A_2}(v')=M_B(v')=1/3$. The latent provider controlling $A_1$ and $A_2$ receives $2/3$, which is strictly larger than $1/2$. Hence $M$ is not false-name-proof. The counterexample is a finite monotone binary-valued game, so any class of games containing it inherits the impossibility.

\subsection{Proof of Proposition~\ref{prop:split-gain}}

In the honest two-player unanimity game $v$, $A$'s only positive marginal contribution is to the singleton $\{B\}$, of size $1$ in a $2$-player game; therefore, $\varphi_A^\omega(v)=\omega_{2,1}$. In the split game $v'$ with $k+1$ players, each pseudonym $A_j$ has a positive marginal contribution only to the coalition of size $k$ that contains all other players (without $A_j$), and the corresponding semivalue weight is $\omega_{k+1,k}$. Hence, $\varphi_{A_j}^\omega(v')=\omega_{k+1,k}$ and the latent provider's total payment is $k\,\omega_{k+1,k}$. The additive split-gain is the difference, and the multiplicative form follows immediately when $\omega_{2,1}\!>\!0$. The Shapley, raw-Banzhaf, and Beta-Shapley specializations follow by substituting weight formulas: Shapley $\omega_{k+1,k}\!=\!1/(k+1)$ and $\omega_{2,1}\!=\!1/2$ give $k/(k+1)-1/2=(k-1)/(2(k+1))$; raw Banzhaf $\omega_{k+1,k}\!=\!2^{-k}$ and $\omega_{2,1}\!=\!1/2$ give $k\,2^{-k}-1/2$, and Beta-Shapley substitution is direct.

\paragraph{$n$-player extension.}
For the $n$-player unanimity game $v_n$ with $v_n([n])\!=\!1$ and zero on proper subsets, every player's Shapley value is $1/n$. After provider $1$ splits into $k\!\geq\!2$ pseudonyms, the resulting $(n\!+\!k\!-\!1)$-player game is again unanimity (value $1$ requires every player), so each pseudonym receives Shapley value $1/(n\!+\!k\!-\!1)$ and the latent provider's total is $k/(n\!+\!k\!-\!1)$. The multiplicative split-gain is therefore $G_{\mathrm{Sh}}(n,k) = nk/(n\!+\!k\!-\!1)$, which exceeds $1$ for all $n,k\!\geq\!2$ and converges to $k$ as $n\!\to\!\infty$. This is the formula instantiated in App.~\ref{app:exp-extra} S1.

\subsection{Proof of Theorem~\ref{thm:exact}}

Let $S$ be an honest profile and $S^\alpha$ be a manipulated profile in which latent provider $i$ splits into pseudonyms $I_i^\alpha$. By assumption, the manipulation does not change the quotient clusters or quotient value function. Therefore for each cluster $C_k$,
\[
    \varphi_k^\omega(\bar v_\theta^\alpha)=\varphi_k^\omega(\bar v_\theta).
\]
The original payment to provider $i$ is
\[
    p_i^{Q,\omega}(S)=\sum_{k=1}^{K}a_{i,k}\varphi_k^\omega(\bar v_\theta).
\]
The total payment to its pseudonyms is
\[
    \sum_{j\in I_i^\alpha}p_j^{Q,\omega}(S^\alpha)
    =\sum_{j\in I_i^\alpha}\sum_{k=1}^{K}a_{j,k}^{\alpha}
      \varphi_k^\omega(\bar v_\theta^\alpha).
\]
Changing the order of summation and using quotient invariance gives
\[
    \sum_{k=1}^{K}\left(\sum_{j\in I_i^\alpha}a_{j,k}^{\alpha}\right)
      \varphi_k^\omega(\bar v_\theta).
\]
By false-name neutrality, the term in parentheses equals $a_{i,k}$ for every $k$. Thus, the total payment of the pseudonyms equals $p_i^{Q,\omega}(S)$, so the additive gain is zero.

\subsection{Proof of Theorem~\ref{thm:approx}}

For $k\in M_i^\alpha$, write $s_k=s_{i,k}^\alpha$, $\varphi_k^\alpha=\varphi_k^\omega(\bar v_\theta^\alpha)$, and $\varphi_h=\varphi_h^\omega(\bar v_\theta)$ for honest clusters. Because $\mu$ is a bijection onto $H_i$, the honest payment terms carrying provider $i$ are used exactly once. The raw estimated payment difference is
\[
\sum_{j\in I_i^\alpha}\widehat p_j^{Q,\omega}(S^\alpha)-\widehat p_i^{Q,\omega}(S)
=\sum_{k\in M_i^\alpha}\left(s_k\widehat\varphi_k^\alpha-a_{i,\mu(k)}\widehat\varphi_{\mu(k)}\right)
+\sum_{k\in E_i^\alpha}s_{i,k}^\alpha\widehat\varphi_k^\alpha.
\]
For a matched pair, add and subtract $s_k\varphi_k^\alpha$, $a_{i,\mu(k)}\varphi_k^\alpha$, and $a_{i,\mu(k)}\varphi_{\mu(k)}$. Since $s_k,a_{i,\mu(k)}\in[0,1]$,
\[
\left|s_k\widehat\varphi_k^\alpha-a_{i,\mu(k)}\widehat\varphi_{\mu(k)}\right|
\leq \left|s_k-a_{i,\mu(k)}\right|\left|\varphi_k^\alpha\right|
+\left|\varphi_k^\alpha-\varphi_{\mu(k)}\right|+2\eta.
\]
For an escaped cluster, $s_{i,k}^\alpha\leq1$ gives $|s_{i,k}^\alpha\widehat\varphi_k^\alpha|\leq|\varphi_k^\alpha|+\eta$. Summing and using $2|M_i^\alpha|+|E_i^\alpha|\leq2K_\alpha$ yields
\[
\left|\sum_{j\in I_i^\alpha}\widehat p_j^{Q,\omega}(S^\alpha)-\widehat p_i^{Q,\omega}(S)\right|
\leq A_i^\alpha+L_i^\alpha+D_i^\alpha+2K_\alpha\eta.
\]
The deterministic false-name-gain bound follows because the signed gain is at most its absolute value. Under the mean-absolute-error premise, take expectations before the final triangle inequality; the same calculation replaces each estimator error by its expected absolute value and gives the stated expectation bound. A union bound over honest and manipulated cluster estimates gives the stated Hoeffding rate.

\subsection{Lemma~\ref{lem:normalized}: normalized-payment bound}

\begin{lemma}[Normalized-payment bound]\label{lem:normalized}
Let $x_i$ and $x_i^\alpha$ be provider $i$'s honest raw aggregate payment and its manipulated aggregate payment, with $|x_i^\alpha-x_i|\leq B_i^\alpha$. Let $T_0,T_\alpha$ be the corresponding total raw-score sums and $V_0,V_\alpha$ the two grand-coalition values. Suppose
\[
T_0,T_\alpha\geq T_{\min}>0,\qquad |T_\alpha-T_0|\leq\Xi^\alpha,\qquad |V_\alpha-V_0|\leq\Delta_V^\alpha,
\]
and $|V_\alpha|\leq V_{\max}$. Then the normalized aggregate payments $\widetilde x_i=V_0x_i/T_0$ and $\widetilde x_i^\alpha=V_\alpha x_i^\alpha/T_\alpha$ satisfy
\[
\left|\widetilde x_i^\alpha-\widetilde x_i\right|
\leq \frac{V_{\max}}{T_{\min}}B_i^\alpha
+\frac{|x_i|}{T_{\min}}\Delta_V^\alpha
+\frac{|V_0|\,|x_i|}{T_{\min}^2}\Xi^\alpha.
\]
\end{lemma}

\begin{proof}
Add and subtract $V_\alpha x_i/T_\alpha$. Then
\[
\left|\frac{V_\alpha x_i^\alpha}{T_\alpha}-\frac{V_0x_i}{T_0}\right|
\leq\frac{|V_\alpha|}{T_\alpha}|x_i^\alpha-x_i|+|x_i|\left|\frac{V_\alpha}{T_\alpha}-\frac{V_0}{T_0}\right|.
\]
The first term is at most $V_{\max}B_i^\alpha/T_{\min}$. For the second,
\[
\left|\frac{V_\alpha}{T_\alpha}-\frac{V_0}{T_0}\right|
\leq\frac{|V_\alpha-V_0|}{T_\alpha}+|V_0|\frac{|T_\alpha-T_0|}{T_\alpha T_0}
\leq\frac{\Delta_V^\alpha}{T_{\min}}+\frac{|V_0|\Xi^\alpha}{T_{\min}^2},
\]
which gives the result. The assumptions on $T_0,T_\alpha,\Xi^\alpha$, and $\Delta_V^\alpha$ are independent stability conditions; Theorem~\ref{thm:approx} alone does not imply them.
\end{proof}

\subsection{Proof of Theorem~\ref{thm:fairness}}

Fix a quotient cluster $k$ and any $Q\subseteq[K]\setminus\{k\}$. By assumption,
\[
    |v^{\mathrm{merge}}(Q\cup\{k\})-\bar v_\theta(Q\cup\{k\})|\leq\Delta_\theta
\]
and
\[
    |v^{\mathrm{merge}}(Q)-\bar v_\theta(Q)|\leq\Delta_\theta.
\]
Therefore, the difference between cluster $k$'s two marginal contributions at coalition $Q$ has absolute value at most $2\Delta_\theta$. Multiplying by the common non-negative weights $\omega_{K,|Q|}$ and summing over $Q$ preserves this bound because $\sum_{Q\subseteq[K]\setminus\{k\}}\omega_{K,|Q|}=1$.

\subsection{Observational ownership ambiguity and within-cluster fairness}

\begin{remark}[Observational ownership ambiguity]\label{rem:ownership-ambiguity}
    A mechanism that observes only submitted data and raw account labels receives the same report when one latent provider controls two observationally identical accounts and when two independent providers submit identical data. The report alone therefore cannot identify which ownership-sensitive allocation is appropriate. This is an identification limitation, not an impossibility between replication resistance and symmetry: a rule can symmetrically divide one cluster payment among identical accounts in both worlds. A stronger entitlement axiom would be needed for a formal incompatibility result.
\end{remark}

\begin{remark}[Within-cluster fairness is a separate object]\label{rem:within-cluster}
    Theorem~\ref{thm:fairness} compares cluster-level totals: player $k$'s semivalue in the merged honest game and the corresponding cluster's semivalue in the quotient game. It does not compare an individual honest unit's semivalue in the original game to the cluster's quotient semivalue. When two or more honest units are merged into one cluster, the original game contains both as players, whereas the quotient game contains only the cluster; an individual-level comparison double-counts the substitution effect. Recovering individual-level fairness inside a cluster therefore requires an additional within-cluster allocation rule (Definition~\ref{def:qsv}) and cannot be derived from quotienting alone.
\end{remark}

%% file: sections/appendix-b-algorithms.tex
\section{Additional algorithms}\label{app:algorithms}

\begin{algorithm}[tbp]
    \caption{Permutation estimator for quotient Shapley}
    \label{alg:perm}
    \begin{algorithmic}[1]
        \Require clusters $\C_\theta=\{C_1,\ldots,C_K\}$; quotient utility $\bar v$; samples $R$
        \State Initialize $\widehat\varphi_k\gets0$ for all $k$
        \For{$r=1$ to $R$}
            \State Draw a uniformly random permutation $\pi$ of $[K]$
            \State $Q\gets\emptyset$; $u_0\gets\bar v(\emptyset)$
            \For{$k$ in order $\pi$}
                \State $u_1\gets\bar v(Q\cup\{k\})$
                \State $\widehat\varphi_k\gets\widehat\varphi_k+(u_1-u_0)/R$
                \State $Q\gets Q\cup\{k\}$; $u_0\gets u_1$
            \EndFor
        \EndFor
        \State \Return $\widehat\varphi_1,\ldots,\widehat\varphi_K$
    \end{algorithmic}
\end{algorithm}

\begin{algorithm}[tbp]
    \caption{Random-subset estimator for quotient Banzhaf}
    \label{alg:banzhaf}
    \begin{algorithmic}[1]
    \Require clusters $\C_\theta=\{C_1,\ldots,C_K\}$; quotient utility $\bar v$; samples $R$
    \For{$k=1$ to $K$}
        \State $\widehat\varphi_k\gets0$
        \For{$r=1$ to $R$}
            \State Draw $Q\subseteq[K]\setminus\{k\}$ by including each cluster independently with probability $1/2$
            \State $\widehat\varphi_k\gets\widehat\varphi_k+\big(\bar v(Q\cup\{k\})-\bar v(Q)\big)/R$
        \EndFor
    \EndFor
    \State \Return $\widehat\varphi_1,\ldots,\widehat\varphi_K$
    \end{algorithmic}
\end{algorithm}

%% file: sections/appendix-c-benchmark.tex
\section{Benchmark implementation details}\label{app:gym}

\textsc{DataMarket-Gym} is the benchmark instantiated for every experiment in Sec.~\ref{sec:experiments}. It exposes the following Python interfaces:
{\small
\begin{verbatim}
ProviderMarket(dataset, provider_partition, learner, utility)
Attack.apply(provider_id, market_state) -> submitted_profile
EvidenceGraph.build(submitted_profile, oracle, threshold) -> clusters
Valuator.fit(clusters, utility_oracle, semivalue, budget) -> cluster_values
Allocator.allocate(cluster_values, provenance_policy) -> payments
Metrics.evaluate(payments, latent_owners, honest_reference) -> report
\end{verbatim}
}

\paragraph{Provider generation.}
Synthetic markets sample $n_{\mathrm{providers}}\!\in\!\{4,6,8\}$ providers, each owning $r/n$ examples drawn i.i.d.\ from a labeled-mixture distribution (synthetic) or a uniformly-random subset of a public benchmark (CIFAR-10, AG~News, STL-10, IMDB). Provider assignment is sequential after a permuted random selection from the training pool, which matches the protocol used in our chaining-floor simulation (Sec.~\ref{sec:exp-real}). Held-out experiments use $n\!=\!4$ providers $\times$ $50$ examples.

\paragraph{Attack library.}
Six attack families, each with a stable name and parameter list:
(i)~\texttt{honest} (no manipulation); (ii)~\texttt{exact\_dup\_2x\_sybils} (provider $0$ duplicates a fraction $f\!\in\![0,1]$ of its units across two pseudonyms); (iii)~\texttt{near\_duplicate\_2x\_sybils}\,(\,$\sigma$\,) (Gaussian noise variants at scale $\sigma\!\in\!\{0.02,0.03\}$ across two pseudonyms; payload-preserving when $\sigma\!=\!0$); (iv)~\texttt{sybil\_split\_k} (pure partition into $k\!\in\!\{2,3,4\}$ pseudonyms with no shared units, round-robin or block); (v)~\texttt{label\_noise} (poisoning baseline; flips $p\!\in\![0,0.5]$ of attacker labels --- out-of-scope for the formal guarantees but included as a stress test); (vi)~\texttt{provider\_zero\_attack} (the named attacker is provider $0$ throughout; baselines invariant under attacker identity).

\paragraph{Evidence graph and representative operator.}
The evidence layer $e_\theta$ is one of \texttt{none} (Shapley over submitted IDs, baseline), \texttt{oracle\_latent} (latent-truth ground truth), \texttt{oracle\_source} (source ID provenance), \texttt{exact\_hash} (perceptual hash to $8$ decimals), \texttt{cosine} (similarity at threshold $\theta$ over precomputed embeddings), \texttt{hybrid\_source\_cosine}, \texttt{hybrid\_hash\_cosine}, or \texttt{noisy\_oracle\_latent}\,(\,$p_{\mathrm{fs}},p_{\mathrm{fm}}$\,) (latent oracle with controlled false-split / false-merge noise; used in S5 grid). Source ID provenance links transformations retaining a common source unit; it does not link a pure split across distinct source units unless the identifier is defined at the provider or license level. The components of the resulting graph are the clusters $\C_\theta$. The representative operator $R_\theta$ defaults to exact-duplicate-collapse (paper Sec.~\ref{sec:quotient}); when duplicates arrive under several identities, count-based allocation additionally needs a fixed canonical-ownership rule rather than assigning the collapsed unit opportunistically. Capped, medoid, and weight-normalized variants are also available and toggled per experiment. The S5 oracle-noise grid uses identity $R_\theta$ to expose Theorem~\ref{thm:approx}'s drift and allocator-leakage channels as $p_{\mathrm{fs}},p_{\mathrm{fm}}$ vary.

\paragraph{Valuation estimator.}
Permutation-sampling Shapley (Algorithm~\ref{alg:perm}) and random-subset Banzhaf (Algorithm~\ref{alg:banzhaf}) over $K$ clusters. Sample budget defaults: $R\!=\!64$ for held-out experiments, $R\!=\!256$ for the S2 main table. Exact enumeration is used when $K\!\leq\!4$ (configurable via \texttt{exact\_n\_limit}). The reported runtime per cell at $4$ providers / $50$ examples is well under one minute on the CPU; the bottleneck is the underlying logistic regression utility fit, not the semivalue sampling. Clustering reduces sampling cost from $r$ raw units to $K$ clusters: at our default ($r\!=\!200, K\!=\!4$ for honest profiles, $K\!=\!6$ for $k\!=\!3$ Sybil attacks), the savings ratio is $\sim\!50\!\times$.

\paragraph{Within-cluster allocation per experiment.}
The S2 main table, IMDB calibration checks, and STL-10 held-out experiment use the \texttt{equal\_submitted} rule (uniform over submitted IDs in a cluster) as the analytic baseline. The S5 noise grid and S7 within-cluster ablation sweep all three rules listed in Table~\ref{tab:allocation-rules} (\texttt{equal\_submitted}, \texttt{count\_canonical}, \texttt{latent\_share}). The S2 Banzhaf columns are raw and explicitly unnormalized; no normalized submitted-identity Banzhaf column is reported. In unanimity games, raw Banzhaf scores remain non-negative while the \emph{additive split-gain} becomes negative for $k\geq3$ (Proposition~\ref{prop:split-gain}).

\paragraph{Utility metric, validation, preprocessing, hyperparameters.}
Utility $U$ is test accuracy on a fixed validation split ($n_{\mathrm{val}}\!=\!500$ for real-domain checks and $200$ for synthetic). The learner is logistic regression with \texttt{max\_iter}$=200$, \texttt{C}$=1.0$, and standardization enabled. Frozen-feature pipelines use \texttt{ResNet-18} (ImageNet pretrained) for images and \texttt{sentence-transformers/all-MiniLM-L6-v2} for text. No fine-tuning; embeddings are precomputed and cached. Random seeds are $\{0,\ldots,9\}$ for IMDB/STL-10, $\{0,\ldots,49\}$ for S2/S4, $\{0,\ldots,29\}$ for S5, and $\{0,\ldots,19\}$ for S7/S8.

\paragraph{Compute.}
All experiments were run on a single CPU node (Apple M-series or a comparable Linux workstation). Total reproduction time for the full paper: synthetic suite $\sim\!10$ min, real-data $\sim\!30$ min, held-out $\sim\!1$ h (including embedding computation). Embeddings were computed using \texttt{scripts/precompute\_*\_embeddings.py} (run automatically by each \texttt{make} target); the resulting \texttt{.npz} caches can be retained locally to bypass the dataset-download stage on subsequent runs.

The reproducibility \texttt{Makefile} exposes six targets that reproduce every empirical claim. They are \texttt{synthetic}, \texttt{real-data}, \texttt{holdout}, \texttt{predict-theta}, \texttt{figures}, and \texttt{test}.

%% file: sections/appendix-d-allocation.tex
\section{Within-cluster allocation rules}\label{app:allocation-rules}

\begin{table}[tbp]
\centering\small
\caption{When each within-cluster allocation rule is false-name-neutral. Failure modes correspond to the empirical rule-divergence in S5/S7 (App.~\ref{app:exp-extra}).}\label{tab:allocation-rules}
\begin{tabular}{@{}p{0.22\linewidth}p{0.36\linewidth}p{0.34\linewidth}@{}}
\toprule
Allocation rule & Neutral when & Fails when \\
\midrule
Equal-share over submitted IDs $a_{j,k}\!=\!1/|I_k|$ & cluster contains pseudonyms of a single latent provider, \emph{or} the split does not change the number of represented IDs in cluster $k$ & mixed clusters: latent provider $i$'s split into $h$ pseudonyms in a cluster with $g$ other IDs raises $i$'s cluster share from $1/(g+1)$ to $h/(g+h)$ \\[1ex]
Count-based over canonical units $a_{j,k}\!\propto\!|R_\theta(C_k)\cap S_j|$ & canonical-unit ownership counts are well defined and invariant under manipulation, using a fixed rule when several identities claim one collapsed duplicate & raw replication when $R_\theta$ is the identity, or ambiguous/adversarial ownership assignment after duplicate collapse \\[1ex]
Count-based over raw submitted units $a_{j,k}\!\propto\!|C_k\cap S_j|$ & pure partitions with conserved raw counts & replication or duplicate-with-Sybil attacks \\[1ex]
Latent-share $a_{j,k}\!\propto\!|\{u\in C_k: \text{latent}(u)\!=\!j\}|$ & reliable latent ownership/provenance signal available to the mechanism & provenance metadata noisy, missing, or attacker-controlled \\
\bottomrule
\end{tabular}
\end{table}

%% file: sections/appendix-e-results.tex
\section{Additional empirical results}\label{app:exp-extra}

This appendix contains the full versions of three synthetic experiments summarized in Sec.~\ref{sec:exp-synthetic}: the closed-form validation of Proposition~\ref{prop:split-gain}, the Monte~Carlo sample-budget ablation, and the oracle-noise grid with its empirical $\Delta_\theta$ vs.\ fairness-loss visualization. All numbers are the mean $\pm$ SE over the seed counts noted in each caption.

\paragraph{S4 extended: Fairness--Sybil threshold frontier.}
Sec.~\ref{sec:exp-synthetic} reports the headline numbers for the cosine-threshold sweep; the full curve and pure-split second axis are below. The frontier is the central empirical object of the paper because it relates the two error sources of Theorems~\ref{thm:approx} and~\ref{thm:fairness}: false splits that allow manipulation gain and false merges that distort honest payments and amplify Sybil shares. On a near-duplicate Sybil with $\sigma\!=\!0.03$, the gain transitions from $1.246\!\pm\!0.052$ at $\theta\!=\!0.85$ to $0.974\!\pm\!0.013$ at $\theta\!=\!0.90$ and saturates at the latent-oracle level $0.964\!\pm\!0.009$ for $\theta\!\geq\!0.95$; the oracle-$L^1$ fairness-loss term traces the same shape ($0.162\!\to\!0.015$). On a pure-split Sybil, the threshold sweep produces no defense at any $\theta$, and at $\theta\!=\!0.85$ the gain is \emph{higher} than baseline Shapley ($1.844\!\pm\!0.054$ vs.\ $1.603$): a poorly-tuned cosine evidence layer is worse than no defense at all because false-merges between independent providers create artificially large clusters that amplify within-cluster Sybil shares, consistent with the within-cluster allocation analysis below. Defending pure splits requires evidence outside the unit-pair graph (Remark~\ref{rem:regimes}); loosening $\theta$ on a content-similarity oracle is an actively counterproductive substitute.
\begin{figure}[tbp]
\centering
\includegraphics[width=0.65\textwidth]{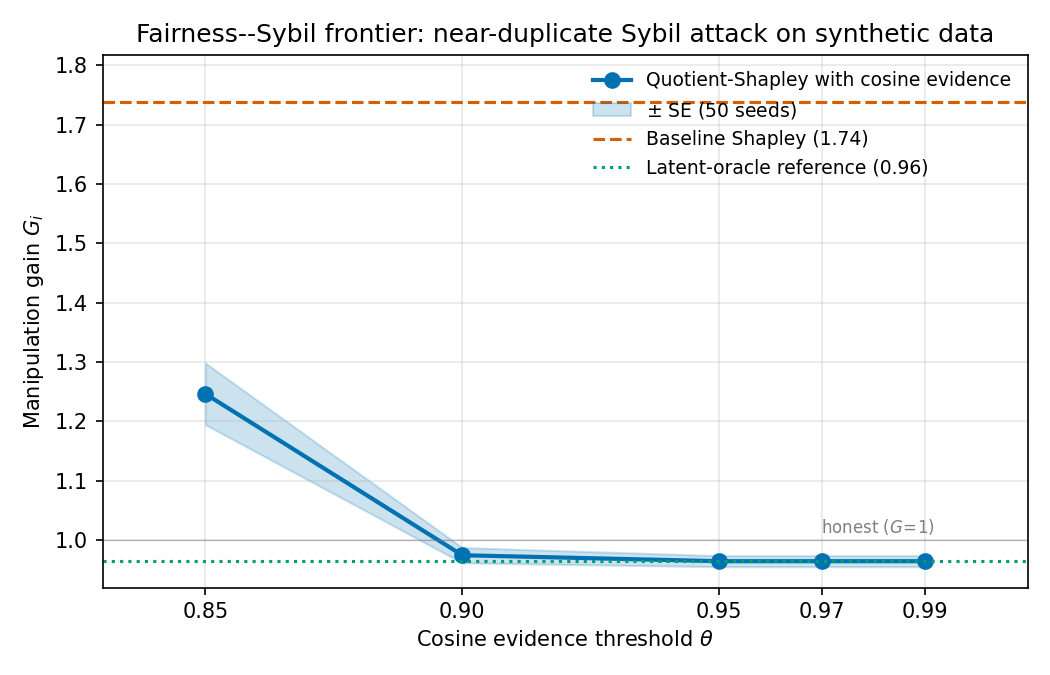}
\caption{Fairness--Sybil frontier: cosine threshold $\theta$ traded against manipulation gain on the near-duplicate-with-Sybil attack. Loose thresholds permit residual gain via false-splits; thresholds at or above $0.90$ saturate to the latent-oracle level. Mean over $50$ seeds; shaded band is $\pm$ SE. Baseline Shapley dashed reference at $1.74$. Note: the latent-oracle dotted reference at $G\!\approx\!0.96$ sits slightly below the honest line at $G\!=\!1$ because the oracle, like any quotient mechanism, accounts for one round of estimator noise; the two references should not be conflated.}\label{fig:frontier}
\end{figure}

\paragraph{S1: Closed-form validation of Proposition~\ref{prop:split-gain}.}
We instantiate unanimity games of size $n\!\in\!\{2,\ldots,6\}$ and apply a pure split of provider~$1$ into $k\!\in\!\{2,\ldots,6\}$ pseudonyms. Each $(n,k)$ cell is computed using exact enumeration. Across $50$ seeds, the measured Shapley multiplicative gain matches $\removed{\Gamma}\new{G}_{\mathrm{Sh}}(n,k)\!=\!nk/(n+k-1)$ to three decimal places (Table~\ref{tab:split-gain-validation}); raw-Banzhaf multiplicative gain matches $G_{\mathrm{Bz}}(k)=k/2^{k-1}$ identically. Random-monotone games show the same qualitative pattern: Shapley has positive additive gain, while raw Banzhaf has negative or zero additive gain for $k\!\geq\!3$ even though its scores remain non-negative.

\begin{table}[tbp]
\centering
\small
\caption{Closed-form multiplicative-gain predictions vs.\ measurements for Proposition~\ref{prop:split-gain} on unanimity games (mean $\pm$ SE over $50$ seeds, averaged across $n\!\in\!\{2,\ldots,6\}$). Predictions: Shapley $G_{\mathrm{Sh}}=nk/(n+k-1)$; raw Banzhaf $G_{\mathrm{Bz}}=k/2^{k-1}$.}\label{tab:split-gain-validation}
\begin{tabular}{llccccc}
\toprule
Mechanism & & $k\!=\!2$ & $k\!=\!3$ & $k\!=\!4$ & $k\!=\!5$ & $k\!=\!6$ \\
\midrule
\multirow{2}{*}{Shapley}     & predicted & $1.563$ & $1.939$ & $2.210$ & $2.417$ & $2.581$ \\
                             & measured  & $1.563{\scriptstyle \pm 0.009}$ & $1.939{\scriptstyle \pm 0.017}$ & $2.210{\scriptstyle \pm 0.024}$ & $2.417{\scriptstyle \pm 0.030}$ & $2.581{\scriptstyle \pm 0.035}$ \\
\midrule
\multirow{2}{*}{raw Banzhaf} & predicted & $1.000$ & $0.750$ & $0.500$ & $0.312$ & $0.187$ \\
                             & measured  & $1.000{\scriptstyle \pm 0.000}$ & $0.750{\scriptstyle \pm 0.000}$ & $0.500{\scriptstyle \pm 0.000}$ & $0.312{\scriptstyle \pm 0.000}$ & $0.187{\scriptstyle \pm 0.000}$ \\
\bottomrule
\end{tabular}
\end{table}

\paragraph{S3: Monte~Carlo sample budget (estimator stability for Theorem~\ref{thm:approx}).}
Sweeping the per-mechanism Shapley sample count $s\!\in\!\{64,128,256,512,1024\}$, the $L^1$ distance between the Shapley estimator's per-provider payments and a high-fidelity reference at $s\!=\!1024$ decays monotonically from $0.200$ at $s\!=\!64$ to $0.149$ at $s\!=\!1024$ on the synthetic classification task, with a clear knee at $s\!=\!256$ (additional samples reduce error by less than $0.02$ at $30\%$ extra runtime).\footnote{This metric measures Shapley estimator stability, not fairness loss to a latent-truth oracle. It is a proxy for the $\eta$ term in Theorem~\ref{thm:approx} only for Shapley-style mechanisms; we restrict S3 to those throughout.} The manipulation gain itself is stable across sample budgets to within $0.04$ ($1.69$--$1.72$ on \texttt{exact\_dup.+sybil}), confirming that $G_i$ is a structural quantity, not an estimator artefact. The latent-oracle quotient mechanism's estimator error is essentially flat at $0.015$ regardless of $s$: perfect clustering reduces effective player count and removes most sampling noise. We use $s\!=\!256$ as the default elsewhere.

\paragraph{S5: Oracle-noise grid (Theorems~\ref{thm:approx} and~\ref{thm:fairness}).}
We instantiate a noisy oracle that takes the latent-truth clustering and corrupts it with two independent error rates: a false-split rate $p_{\mathrm{fs}}$ (each true edge is removed with probability $p_{\mathrm{fs}}$) and a false-merge rate $p_{\mathrm{fm}}$ (each non-edge between distinct providers is added with probability $p_{\mathrm{fm}}$). Both rates sweep $\{0, 0.05, 0.10, 0.20, 0.40\}$, giving a $5\!\times\!5$ grid evaluated on the near-duplicate Sybil attack with $30$ seeds. Two patterns emerge under equal-share within-cluster allocation (Table~\ref{tab:noise-grid}). First, manipulation gain rises monotonically along both axes, but \emph{false-merge rate dominates}: at $p_{\mathrm{fs}}\!=\!0$, increasing $p_{\mathrm{fm}}$ from $0$ to $0.40$ pushes gain from $0.96$ to $2.46$ (worse than baseline Shapley); at $p_{\mathrm{fm}}\!=\!0$, increasing $p_{\mathrm{fs}}$ from $0$ to $0.40$ only moves gain to $1.13$. False merges create mixed clusters in which equal-share violates Assumption~\ref{ass:false-name-neutral}; the resulting increase is captured by $A_i^\alpha>0$ in Theorem~\ref{thm:approx}. In contrast, false splits primarily create escaped mass $L_i^\alpha$. Second, the cluster-level fairness-loss term (oracle $L^1$ distance) tracks $p_{\mathrm{fm}}$ almost independently of $p_{\mathrm{fs}}$ ($0.015\!\to\!0.32$ as $p_{\mathrm{fm}}\!:\!0\!\to\!0.40$), validating the form of Theorem~\ref{thm:fairness}'s bound. The asymmetry is rule-specific (S7 below).

\begin{table}[tbp]
\centering
\small
\caption{Manipulation gain on near-duplicate Sybil attacks under a noisy provenance oracle. Rows: false-split rate $p_{\mathrm{fs}}$. Columns: false-merge rate $p_{\mathrm{fm}}$. Cell $(0,0)$ is the clean oracle. Bold cells exceed baseline Shapley ($G\!=\!1.71$ on this attack): the defense is worse than no defense.}\label{tab:noise-grid}
\begin{tabular}{l|ccccc}
\toprule
$p_{\mathrm{fs}}\!\setminus\!p_{\mathrm{fm}}$ & $0$ & $0.05$ & $0.10$ & $0.20$ & $0.40$ \\
\midrule
$0$ & $0.958$ & $1.202$ & $1.426$ & $\mathbf{1.886}$ & $\mathbf{2.457}$ \\
$0.05$ & $0.958$ & $1.215$ & $1.457$ & $\mathbf{1.983}$ & $\mathbf{2.460}$ \\
$0.10$ & $0.972$ & $1.242$ & $1.460$ & $\mathbf{1.995}$ & $\mathbf{2.441}$ \\
$0.20$ & $1.015$ & $1.248$ & $1.473$ & $\mathbf{1.981}$ & $\mathbf{2.457}$ \\
$0.40$ & $1.131$ & $1.356$ & $1.587$ & $\mathbf{2.028}$ & $\mathbf{2.497}$ \\
\bottomrule
\end{tabular}
\end{table}

\paragraph{S6: Empirical $\Delta_\theta$ vs.\ fairness loss.}
Theorem~\ref{thm:fairness} bounds cluster-level fairness loss by $2\Delta_\theta$. Using the mixed-component fraction (the share of attribution clusters that mix multiple latent owners) as an empirical $\Delta_\theta$ proxy and plotting it against oracle $L^1$ payment distance across S4 (cosine threshold) and S5 (noise grid) on the near-duplicate Sybil attack (Figure~\ref{fig:empirical-delta-theta}), measured loss stays well below the linear bound $L^1\!\leq\!2\Delta_\theta$ in every cell, rising with cluster mixing but saturating near $0.4$ as $\Delta_\theta\!\to\!1$. Theorem~\ref{thm:fairness}'s bound is therefore conservative on this DGP.

\begin{figure}[tbp]
\centering
\includegraphics[width=0.65\textwidth]{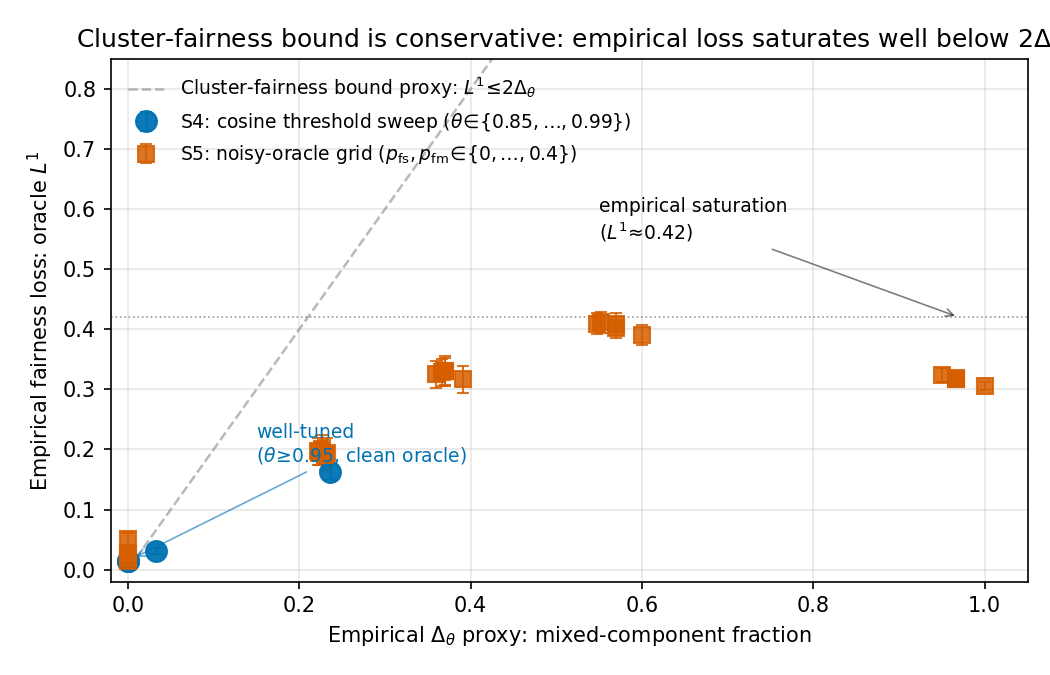}
\caption{Empirical $\Delta_\theta$ proxy (mixed-component fraction) vs.\ empirical fairness loss (oracle $L^1$ error) across S4 and S5, near-duplicate Sybil attack.}\label{fig:empirical-delta-theta}
\end{figure}

\paragraph{S7: Within-cluster allocation rule changes the false-merge channel.}
S5 assumes equal-share allocation. We compare it with \emph{count-based} allocation over canonical ownership counts and \emph{latent-share} allocation. Re-running the noise grid with all three rules at twenty seeds reverses the headline (Figure~\ref{fig:allocation-sensitivity}). At $p_{\mathrm{fm}}\!=\!0.20,p_{\mathrm{fs}}\!=\!0$ on a near-duplicate Sybil, equal-share yields $G\!=\!1.80\!\pm\!0.04$; count-based and latent-share both yield $G\!=\!1.38\!\pm\!0.04$, a $23\%$ reduction. On pure Sybil, equal-share increases to $1.82$ but count-based decreases to $0.84$, below the honest level. The two non-equal rules coincide numerically because each submitted identity in this DGP maps cleanly to one latent owner, making the relevant ownership counts well defined and invariant. The result favors count-based allocation under those conditions; it does not establish a universal production default.

\begin{figure}[tbp]
\centering
\includegraphics[width=0.95\textwidth]{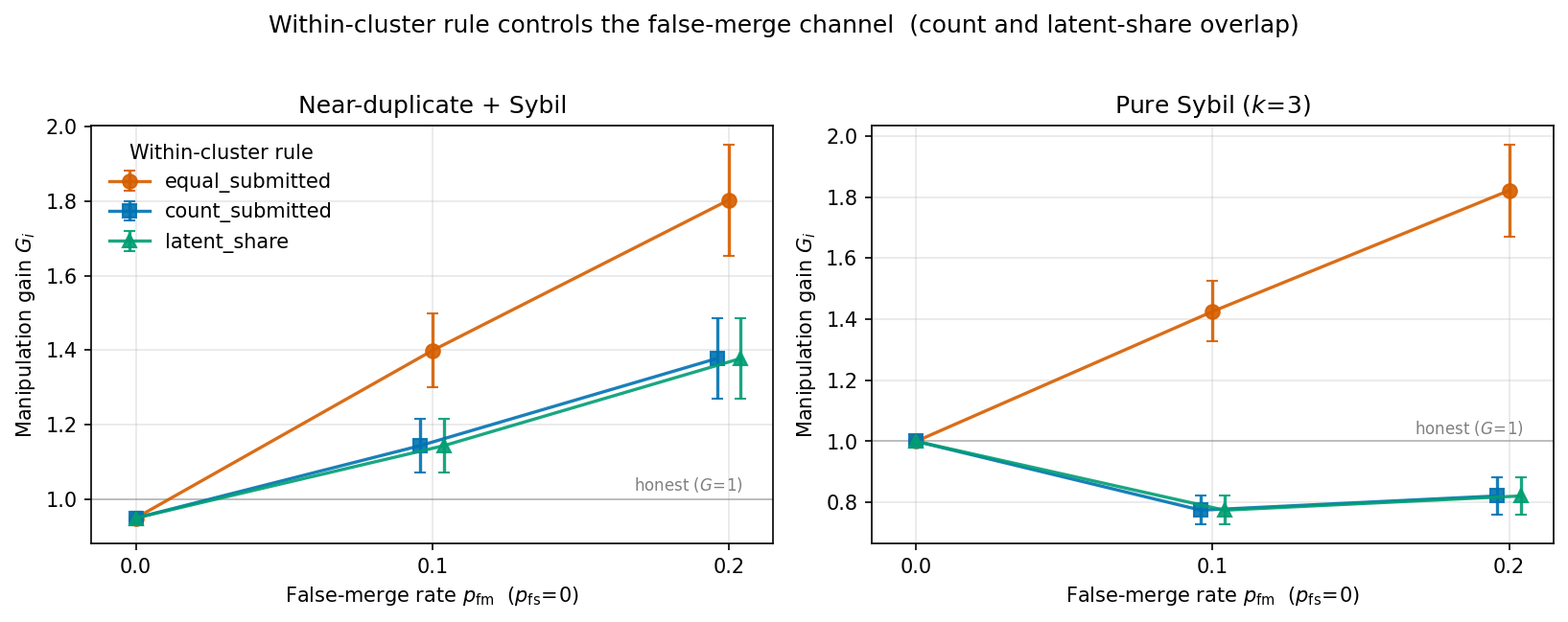}
\caption{Manipulation gain on near-duplicate Sybil and pure Sybil $k\!=\!3$ under three within-cluster allocation rules, as $p_{\mathrm{fm}}$ varies with $p_{\mathrm{fs}}\!=\!0$. Right panel: count- and latent-share rules push pure-Sybil gain \emph{below} the honest level ($G\!<\!1$) at $p_{\mathrm{fm}}\!>\!0$ because false-merges with attacker pseudonyms dilute the latent provider's share within mixed clusters.}\label{fig:allocation-sensitivity}
\end{figure}

\paragraph{S8: DGP robustness across $n_{\mathrm{providers}}$ and class separation.}
We sweep $n_{\mathrm{providers}}\!\in\!\{6,8\}$ and class-separation $\in\!\{0.6,1.2,1.8\}$ on the synthetic classification task, holding $20$ seeds and the four core attacks fixed. Baseline-Shapley gain on \texttt{exact\_dup\_2x\_sybils} stays in $[1.69,1.81]$ across all six cells; quotient-Shapley with the latent oracle stays in $[0.92,0.99]$; quotient-Shapley with cosine evidence at $\theta\!=\!0.99$ matches the latent-oracle level within $\pm 0.01$ in every cell. On \texttt{sybil\_k3} the oracle-latent quotient is exactly $1.000$ across every cell, while baseline Shapley sits in $[1.50,1.76]$. The headline finding (quotient with example-level evidence eliminates the replication-Sybil channel; only the latent oracle handles pure splits) is therefore robust to provider count and coalition complementarity in this regime. Cells at $n_{\mathrm{providers}}\!=\!12$ exceeded the $2$-hour task walltime owing to exact-enumeration cost and are not reported; sampling-based estimation would extend the grid further.

\paragraph{Embedding geometry behind the threshold reversal.}
Table~\ref{tab:embed-geom} reports the three domain constants that determine the admissible cosine threshold (Sec.~\ref{sec:exp-real}). Statistics are computed on the full training pool ($50{,}000$ images, $20{,}000$ AG News articles), with near-duplicate cosines from $5{,}000$ sampled vectors against $\mathcal{N}(0,\sigma^2 I)$ noise at $\sigma\!=\!0.02$ and the chaining floor from a $30$-trial provider-level mixed-component-fraction simulation matching the experimental graph protocol. The pairwise floor (intra-class cosine p90) is binding when feature norms are sufficiently large, resulting in a low cross-provider density at the working $\theta$; the chaining floor is binding when unit-normalized embeddings produce sufficient cross-provider density to collapse the provider-level evidence graph at sub-threshold edges. The near-duplicate p10 is the universal upper bound. The script that computes both floors and the upper bound from cached embeddings is \texttt{scripts/predict\_theta.py}.

\paragraph{Why two floors? Pairwise vs chaining intuitively.}
The lower bound on $\theta$ has to defeat two distinct over-merge mechanisms. \emph{Pairwise over-merge} fails locally: a single same-class honest pair has cosine $\geq\!\theta$ and gets glued into one cluster, distorting the cluster total. The fix is $\theta\!>$ the typical honest same-class pair cosine, i.e.\ intra-class p90. \emph{Chaining} fails globally: each individual cross-provider pair may have cosine just below $\theta$, but with $r/n$ units per provider and a non-trivial cross-provider edge density, transitive closure of the evidence graph can still link all $n$ providers into one giant component. The fix is $\theta\!>$ the value at which the simulated cross-provider edge density drops low enough that the provider-level graph stays disconnected --- operationalized as the smallest $\theta$ at which the simulated provider-level MCF falls below $0.10$. The MCF$<\!0.10$ cutoff is calibrated against the original IMDB attack-defense sweep at $\theta\!\in\!\{0.50, 0.70, 0.85, 0.90, 0.95\}$. Sensitivity to the cutoff choice (Table~\ref{tab:cutoff-sens}): the MCF curves transition sharply enough that predicted floors are stable to within $\pm 0.03$ across cutoff $\in\![0.05, 0.20]$ on CIFAR/STL-10/IMDB and $\pm 0.06$ on AG~News, with the $0.10$ central choice giving $0.88$ (CIFAR), $0.78$ (STL-10), $0.84$ (AG~News), $0.78$ (IMDB).

\begin{table}[tbp]
\centering\small
\caption{Chaining floor as a function of the MCF cutoff. Sharp curve transitions keep the calibrated rule stable across development, calibration, and held-out domains.}\label{tab:cutoff-sens}
\begin{tabular}{lcccc}
\toprule
Domain & cutoff $0.05$ & cutoff $0.10$ & cutoff $0.15$ & cutoff $0.20$ \\
\midrule
CIFAR-10 (ResNet-18) & $0.88$ & $\mathbf{0.88}$ & $0.86$ & $0.86$ \\
STL-10 (ResNet-18)   & $0.80$ & $\mathbf{0.78}$ & $0.78$ & $0.78$ \\
AG~News (MiniLM)     & $0.88$ & $\mathbf{0.84}$ & $0.80$ & $0.76$ \\
IMDB (MiniLM)        & $0.80$ & $\mathbf{0.78}$ & $0.76$ & $0.74$ \\
\bottomrule
\end{tabular}
\end{table} In all four of our domains chaining binds (CIFAR/AG~News/IMDB strictly, STL-10 in a tie with the pairwise floor). The simulation matches the experimental task partition --- class-stratified providers (each provider holds one class) for image tasks per \texttt{tasks/cifar.py}, and random partitioning for text tasks per \texttt{tasks/text.py} --- which is essential: under random partition CIFAR's chaining floor is artificially inflated by intra-class cross-provider pairs that the experiment never produces.
\begin{table}[tbp]
\centering
\small
\caption{Embedding-geometry quantities driving the cosine-threshold reversal. CIFAR-10 and AG~News are development domains; IMDB calibrates the MCF cutoff and supplies retrospective checks; STL-10 is the held-out validation domain. Chaining floor: smallest $\theta$ at which the simulated provider-level mixed-component fraction (MCF) drops below $0.10$.}\label{tab:embed-geom}
\resizebox{\textwidth}{!}{%
\begin{tabular}{lcccccc}
\toprule
Domain & $\|x\|$ (med) & pairwise floor & chaining floor & near-dup p10 & binding floor & admissible $\theta$ \\
\midrule
\multicolumn{7}{l}{\emph{Development domains (rule construction)}}\\
CIFAR-10 (ResNet-18, 512-d) & $27.9$ & $0.81$ & $0.88$ & $0.9998$ & $0.88$ (chaining) & $[0.88,\,1.00]$ \\
AG~News (MiniLM, 384-d)     & $1.00$ & $0.26$ & $0.84$ & $0.9246$ & $0.84$ (chaining) & $[0.84,\,0.92]$ \\
\midrule
\multicolumn{7}{l}{\emph{Calibration domain (retrospective sharpness check)}}\\
IMDB (MiniLM, 384-d)        & $1.00$ & $0.44$ & $0.78$ & $0.9244$ & $0.78$ (chaining) & $[0.78,\,0.92]$ \\
\midrule
\multicolumn{7}{l}{\emph{Held-out validation domain}}\\
STL-10 (ResNet-18, 512-d)   & $27.1$ & $0.78$ & $0.78$ & $0.9998$ & $0.78$ (tie) & $[0.78,\,1.00]$ \\
\bottomrule
\end{tabular}%
}
\end{table}

\paragraph{Calibration and held-out manipulation gain at each tested $\theta$.}
Table~\ref{tab:holdout} reports the cosine-threshold sweep on a near-duplicate Sybil attack for the IMDB calibration check and STL-10 held-out validation. Bold cells fall inside the admissible interval in Table~\ref{tab:embed-geom}. IMDB's added cells retrospectively locate the chaining transition around $0.78$; STL-10 tests the rule out of domain and recovers near-oracle gain above its predicted floor.
\begin{table}[tbp]
\centering
\small
\caption{Manipulation gain $G_i$ on the near-duplicate Sybil attack ($10$ seeds, mean; SE $\leq\!0.16$): IMDB is a calibration-domain check and STL-10 is held out. Bold cells fall inside the admissible interval in Table~\ref{tab:embed-geom}; em-dash entries were not run.}\label{tab:holdout}
\resizebox{\textwidth}{!}{%
\begin{tabular}{l@{\;\;}ccccccccccc@{\;}cc}
\toprule
& \multicolumn{11}{c}{cosine-threshold $\theta$} & oracle & baseline\\
\cmidrule(lr){2-12}
& $0.30$ & $0.50$ & $0.65$ & $0.70$ & $0.75$ & $0.78$ & $0.80$ & $0.82$ & $0.85$ & $0.90$ & $0.95$ & latent & Shap.\\
\midrule
IMDB (MiniLM, $[0.78,0.92]$)
& --- & $2.02$ & --- & $1.38$ & $1.20$ & $\mathbf{1.01}$ & $\mathbf{1.01}$ & $\mathbf{1.01}$ & $\mathbf{1.01}$ & $\mathbf{1.01}$ & $1.92$ & $1.01$ & $1.90$\\
STL-10 (ResNet-18, $[0.78,1.00]$)
& $2.00$ & $2.00$ & $2.00$ & $1.95$ & $1.70$ & --- & $\mathbf{1.00}$ & --- & $\mathbf{1.00}$ & $\mathbf{1.00}$ & $\mathbf{1.00}$ & $1.00$ & $1.00$\\
\bottomrule
\end{tabular}%
}
\end{table}

\paragraph{S9: Exploratory scale check on text at $n\!=\!8$, $100$ examples per provider.} One AG~News seed at $n_{\mathrm{providers}}\!=\!8$ and $100$ examples each took $2$h\,$19$m; the corresponding CIFAR run exceeded the $4$h walltime and is incomplete. In the completed text run, baseline Shapley gives $G\!=\!3.30$ on near-duplicate Sybil and $3.96$ on pure Sybil $k\!=\!3$, while latent-oracle quotient Shapley gives $0.99$ and $1.00$. Cosine $\theta\!\leq\!0.90$ gives $0.99$, whereas $\theta\!\geq\!0.95$ gives $4.08$. This single seed is compatible with the four-provider pattern but cannot establish scale robustness or cross-domain transfer; it is exploratory and is not used as evidence for the conclusion.

\paragraph{Planned ablations.} Beyond what is reported above, the benchmark supports: (i)~oracle-quality grids beyond the $5\!\times\!5$ used for S5; (ii)~quotient-Banzhaf and quotient-Beta-Shapley sweeps under the same evidence layers, to confirm that quotienting (not the semivalue) is the main source of Sybil resistance; (iii)~provider-granularity comparisons (example-level, batch-level, provider-level units); (iv)~negative-contribution attacks with mislabeled or poisoned data, evaluating whether non-negative payment constraints hide harms.

%% file: sections/appendix-f-reproducibility.tex
\section{Reproducibility checklist notes}\label{app:reproducibility}

\begin{itemize}[leftmargin=*]
    \item \textbf{Code:} Reproducibility artifact at \url{https://anonymous.4open.science/r/neurips-2026-quotient-semivalues-artifact-7F63/}. The artifact contains the \texttt{datamarket-gym} benchmark, all experiment configs (synthetic, CIFAR-10/AG~News development domains, IMDB calibration checks, STL-10 held-out validation), embedding pre-computation and geometry-analysis scripts, unit tests, a pinned \texttt{uv} environment lockfile, and \texttt{Makefile} targets \texttt{synthetic}, \texttt{real-data}, \texttt{holdout}, \texttt{predict-theta}, \texttt{figures}, \texttt{test} that reproduce every empirical claim in the paper.
    \item \textbf{Data:} All experiments use public benchmarks whose license and usage terms were checked before artifact release.
    \item \textbf{Compute:} Appendix~\ref{app:gym} reports hardware, valuation sample counts, and total wall-clock time.
    \item \textbf{Randomness:} Experiments use fixed seeds for partitioning, attacks, training, and valuation; Appendix~\ref{app:gym} lists them.
    \item \textbf{Negative results:} Sections~\ref{sec:exp-synthetic}--\ref{sec:exp-real} report thresholds where quotienting over-merges independent contributors.
\end{itemize}